\pdfoutput=1
\documentclass[11pt,letterpaper]{amsart}

\usepackage[left=1in,right=1in,top=1in,bottom=1in]{geometry}
\usepackage[foot]{amsaddr}
\usepackage{xspace}
\usepackage{mathscinet}
\usepackage{amsmath}
\usepackage{amsthm}
\usepackage{amssymb}
\usepackage{thmtools}
\usepackage{url}
\usepackage{tikz}  
\usepackage{hyperref}
\usepackage{todonotes}
\usepackage{mathtools}
\usepackage{accents}
\pdfoutput=1
\pgfdeclarelayer{bg}
\pgfsetlayers{bg,main}
\newtheorem{Thm}{Theorem}[section]
\newtheorem{Lem}[Thm]{Lemma}

\newtheorem{Obs}[Thm]{Observation}
\theoremstyle{remark}

\theoremstyle{definition}
\newtheorem{Def}[Thm]{Definition}

\newenvironment{Proof}{\begin{proof}}{\end{proof}}

\newcommand{\pw}{\ensuremath{\mathtt{pw}}\xspace}
\newcommand{\tw}{\ensuremath{\mathtt{tw}}\xspace}
\newcommand{\treedecomp}{\ensuremath{\mathbb{T}}\xspace}
\newcommand{\rootv}{\ensuremath{r}\xspace}

\DeclareRobustCommand{\deg}{\operatorname{deg}}

\newcommand{\F}{\ensuremath{\mathbb{F}}}
\newcommand{\R}{\ensuremath{\mathbb{R}}}

\newcommand{\Ra}{\mathbf{R}}
\newcommand{\AR}{\underaccent{\sim}{\mathbf{R}}}

\DeclareRobustCommand{\perm}{\operatorname{perm}}
\DeclareRobustCommand{\haf}{\operatorname{haf}}
\DeclareRobustCommand{\supp}{\operatorname{supp}}
\DeclareRobustCommand{\osupp}{\operatorname{osupp}}
\DeclareRobustCommand{\diag}{\operatorname{diag}}

\begin{document}

%%%%%%%%%%%%%%%%%%%%%%%%%%%%%%%%%%%%%%%%%%%%%%%%%% Front-matter information %%%

\title[]{Kronecker scaling of tensors with applications to arithmetic circuits and algorithms}

% Comment the line below for authors and acknowledgments
%\def\A{Anonymous}

\ifdefined\A
\else
\author{Andreas Bj\"orklund}
\address{IT University of Copenhagen}
\email{andreas.bjorklund@yahoo.se}

\author{Petteri Kaski}
\address{Aalto University}
\email{petteri.kaski@aalto.fi}

\author{Tomohiro Koana}
\address{Kyoto Univeristy}
\email{tomohiro.koana@gmail.com}

\author{Jesper Nederlof}
\address{Utrecht University}
\email{j.nederlof@uu.nl}
\fi

\begin{abstract}
We show that sufficiently low tensor rank for the balanced tripartitioning tensor 
$P_d(x,y,z)=\sum_{A,B,C\in\binom{[3d]}{d}:A\cup B\cup C=[3d]}x_Ay_Bz_C$
for a large enough constant $d$ implies uniform arithmetic circuits
for the matrix permanent that are exponentially smaller than circuits 
obtainable from Ryser's formula.

We show that the same low-rank assumption implies exponential time improvements over the state of the art for a wide variety of other related counting and decision problems.

As our main methodological contribution, we show that the tensors $P_n$ have a desirable \emph{Kronecker scaling} property: They can be decomposed efficiently into a small sum of restrictions of Kronecker powers of $P_d$ for constant $d$.
We prove this with a new technique relying on Steinitz's lemma, which we hence call \emph{Steinitz balancing}.

As a consequence of our methods, we show that the mentioned low rank assumption (and hence the improved algorithms) is implied by Strassen's asymptotic rank conjecture [{\em Progr.~Math.}~120 (1994)], a bold conjecture that has recently seen intriguing progress.
\end{abstract}

\maketitle
%%%%%%%%%%%%%%%%%%%%%%%%%%%%%%%%%%%%%%%%%%%%%%%%%%%%%%%%%%%%% Document body %%%
\thispagestyle{empty} 
\clearpage
\setcounter{page}{1}
\section{Introduction}

\label{sect:introduction}

Tensors, or, equivalently, set-multilinear polynomials, are among the key
objects of interest in the study of arithmetic circuits, algorithms, and
complexity. As was discovered and advanced by Strassen during the course of 
the 1970s and 1980s~\cite{Strassen1969,Strassen1973,Strassen1986,Strassen1987,Strassen1988,Strassen1991,Strassen1994,Strassen2005}---Wigderson and Zuiddam~\cite{WigdersonZ2023} give a recent broad overview---already the study of 
three-way/trilinear tensors gives rise to a deep theory of bilinear complexity 
capturing fundamental computational 
problems such as the task of multiplying two given matrices, with 
substantial connections to algebraic geometry~(e.g.~\cite{BiniCRL1979,buczynska2021apolarity,buczynski2013ranks,landsberg2010ranks,Landsberg2012,Landsberg2019,Schonhage1981,zak1993tangents}) as well as 
more recent connections~\cite{BjorklundCHKP2025,BjorklundK2024,Pratt2024} 
to aspects of fine-grained complexity theory and problems that are 
{\em a priori} perhaps of a more combinatorial nature, such as 
the Set Cover conjecture~\cite{CyganDLMNOPSW2016,CyganFKLMPPS15,KrauthgamerT2019} and the chromatic number problem on graphs~\cite{BjorklundCHKP2025}.
Central to Strassen's theory is to understand properties of {\em sequences} of 
three-tensors of the {\em Kronecker power} form
\begin{equation}
\label{eq:power-seq}
S\,,\quad S^{\otimes 2}\,,\quad S^{\otimes 3}\,,\quad\ldots
\end{equation}
for some constant-size tensor $S$ over a field $\mathbb{F}$, such as the
$4\times 4\times 4$ tensor $\mathrm{MM}_2$ that represents $2\times 2$ matrix 
multiplication as a bilinear map in coordinates. In essence, each tensor in 
the sequence \eqref{eq:power-seq} is ``smooth'' in the sense that it factors 
into a Kronecker power of the generator tensor $S$.\footnote{Of particular interest in Strassen's theory 
is to understand the exponential rate of growth of the tensor 
rank $\Ra(S^{\otimes q})$ along the sequence \eqref{eq:power-seq}, 
formalized as the {\em asymptotic rank} 
$\AR(S)=\lim_{q\rightarrow\infty}\Ra(S^{\otimes q})^{1/q}$ 
of $S$~\cite{Gartenberg1985}. 
For example, Strassen showed~\cite{Strassen1986,Strassen1988} that
the asymptotic rank of the tensor $\mathrm{MM}_2$ captures the 
exponent $\omega$ of square matrix multiplication 
by $\AR(\mathrm{MM}_2)=2^\omega$. 
Also, Strassen showed~\cite{Strassen1988} 
(see also \cite{ChristandlVZ2021,ConnerGLV2022})
that for an arbitrary $d\times d\times d$ tensor
$S$ it holds that $\AR(S)\leq d^{2\omega/3}$; that is, unlike matrix rank,
tensor rank for generic three-tensors is strictly submultiplicative when
taking of Kronecker powers. We postpone our standard notational conventions 
with tensors to Section~\ref{sect:preliminaries}.}{}

While Strassen's theory has been highly successful in advancing algebraic 
complexity in the domain of polynomial-complexity problems such as problems in the matrix-multiplication family~(e.g.~\cite{BurgisserCS1997}), strong connections between Strassen's trilinear theory and the algebraic complexity of conjectured canonical hard problems and algebraic complexity classes have yet been lacking. Most notably so in the case of Valiant's theory of VNP-completeness~\cite{Valiant1979b} and the study of the {\em matrix permanent}, which is also the canonical \#P-complete problem~\cite{Valiant1979} in Valiant's theory of counting complexity. Mulmuley's geometric complexity theory~\cite{Mulmuley2012,Mulmuley2011} and Raz's~\cite{Raz2013} seminal study connecting arithmetic formula complexity and tensor rank of higher-order tensors signal that techniques from algebraic geometry and the study of tensor rank should have further a say here, as do the recent techniques~\cite{BjorklundCHKP2025,BjorklundK2024,Pratt2024} connecting the fine-grained study of {\em combinatorial} NP-complete problems to Strassen's trilinear theory. This suggests a stronger connection between Strassen's trilinear theory and the theory of {\em arithmetic} circuits for hard problems could be made, in particular, if one could push the envelope on analysis of tensor sequences in Strassen's trilinear theory beyond strict Kronecker power sequences~\eqref{eq:power-seq} and asymptotic rank.
This paper shows that such an analysis is possible and such connections exist
between Strassen's trilinear theory and Valiant's theories of algebraic 
complexity and counting complexity. 

\subsection{The Kronecker scaling property and exponents for sequences of tensors}

In this paper, we expand the reach of Strassen's trilinear theory to sequences
of three-tensors 
\begin{equation}
\label{eq:scaling-seq}
T_1\,,\quad T_2\,,\quad T_3\,,\quad\ldots
\end{equation}
that are {\em not}\/ Kronecker powers~\eqref{eq:power-seq}, but have 
a {\em Kronecker scaling} property 
of ``approximate'' smoothness in the following precise sense: 
\begin{quote}
For all $\delta>0$ there exist infinitely many $d=1,2,\ldots$ such that 
for all large enough $n=1,2,\ldots$ in an arithmetic progression
the tensor $T_n$ is a sum of 
at most $2^{\delta n}$ tensors, each of which is a restriction 
of $T_d^{\otimes s}$ for $s\leq (1+\delta)n/d$.
\end{quote}
It is immediate that a Kronecker power sequence \eqref{eq:power-seq} has 
the Kronecker scaling property. 
What is considerably less immediate---and our main result in this paper---is 
that the sequence of {\em balanced tripartitioning} tensors
has this property. In what follows we view three-tensors as 
set-multilinear polynomials in three sets of indeterminates $x,y,z$
and write $[k]=\{1,2,\ldots,k\}$.
\begin{restatable}[Main; Kronecker scaling for balanced tripartitioning tensors]{Thm}{kroneckerscaling}
\label{thm:main-kronecker-scaling}
The sequence of balanced tripartitioning tensors 
\begin{equation}
\label{eq:balanced-tripartitioning}
P_n=P_n(x,y,z)=\sum_{\substack{A,B,C\in\binom{[3n]}{n}\\A\cup B\cup C=[3n]}}x_Ay_Bz_C\qquad\text{for $n=1,2,\ldots$}
\end{equation}
has the Kronecker scaling property.
\end{restatable}

Combined with the fact that the decomposition 
underlying Theorem~\ref{thm:main-kronecker-scaling} is efficiently computable 
in a sense to be made precise later, and using well-known techniques in 
Strassen's trilinear theory, our main result has the following corollary 
in terms of uniform arithmetic circuits:

\begin{restatable}[Uniform circuits for balanced tripartitioning polynomials]{Thm}{tripartitioncircuit}
\label{thm:uniform-circuits-for-balanced-tripartitioning}
Let $\Lambda\geq 1$ be a constant such that the tensor rank of $P_d$ satisfies 
$\Ra(P_d)\leq \Lambda^d$ for all large enough $d$.
Then, for all $\Gamma>\Lambda$ it holds that there exists
an algorithm that given $n$ as input in time $O(\Gamma^n)$ constructs 
an arithmetic circuit of size $O(\Gamma^n)$ for the polynomial $P_n(x,y,z)$. 
\end{restatable}

Theorem~\ref{thm:uniform-circuits-for-balanced-tripartitioning}
highlights the significance of the exponential rate of growth $\Lambda^d$
of the tensor rank $\Ra(P_d)$ along the 
sequence~\eqref{eq:balanced-tripartitioning} as $d$ grows.\footnote{Here it is perhaps worthwhile to stress the quantity of interest
is the tensor rank, {\em not} the asymptotic rank, along the sequence \eqref{eq:balanced-tripartitioning}. Later in Theorem~\ref{thm:asymptotic-scaling} we will, however, show that it is a nontrivial consequence of the Kronecker
scaling property and our main theorem 
(Theorem~\ref{thm:main-kronecker-scaling} and its explicit-decomposition
version, Theorem~\ref{thm:kronecker-scaling}) that the rank 
and asymptotic rank have identical exponents along the sequence.}{}
It will be convenient to study such growth rates via 
{\em exponents} of three-tensor sequences as recently studied 
in~\cite{KaskiM2025}; for a sequence $T_{\mathbb{N}}$ 
consisting of three-tensors $T_n$ of shape $s_n\times s_n\times s_n$ 
for $n\in\mathbb{N}$, define the {\em exponent} 
\begin{equation}
\label{eq:exponent}
\sigma(T_\mathbb{N})=\inf\,\{\sigma>0:\Ra(T_n)\leq s_n^{\sigma+o(1)}\}\,.
\end{equation}
When $S$ is an individual tensor of shape $d\times d\times d$, 
we write $\sigma(S)$ for the exponent of the Kronecker power 
sequence~\eqref{eq:power-seq}; the exponent $\sigma(S)$ and the 
asymptotic rank $\AR(S)$ are related by $\AR(S)=d^{\sigma(S)}$.
The exponent~\eqref{eq:exponent} has the convenience that it abstracts away 
the shape of the underlying tensors and thus enables more concise 
complexity characterizations. For example, using exponents, Strassen's 
characterization~\cite{Strassen1986,Strassen1988} of the 
exponent $\omega$ of square matrix multiplication 
becomes $\omega=2\sigma(\mathrm{MM}_2)$.

Analogously to the matrix multiplication exponent $\omega$, in the language 
of exponents, our applications presented in the next subsection motivate the 
following question concerning the sequence $P_\mathbb{N}$ consisting of 
the balanced tripartitioning 
tensors \eqref{eq:balanced-tripartitioning}:
\medskip
\begin{quote}
What is the value of the exponent $\sigma(P_\mathbb{N})$ of balanced
tripartitioning?
\end{quote}
\medskip
We know that the exponent $\sigma(P_\mathbb{N})$ satisfies 
$1\leq\sigma(P_\mathbb{N}) \leq H(1/3)^{-1}$, where 
$H(\lambda)=-\lambda\log_2\lambda-(1-\lambda)\log_2(1-\lambda)$ is
the binary entropy function. 
\footnote{%
Indeed, here the lower bound follows from matrix
rank by a standard flattening argument for the tensor $P_n$, and the
upper bound is a consequence of Stirling's formula (e.g.~\cite{Robbins1955}) 
and the fact that $P_n$ is a restriction of 
shape $\binom{3n}{n}\times \binom{3n}{n}\times \binom{3n}{n}$ 
of the $(3n)$\textsuperscript{th} Kronecker power of a 
tensor of shape $2\times 2\times 2$, where all the latter tensors are known 
to have border rank at most $2$ (e.g.~\cite{Landsberg2012}).}{}
As our applications will motivate, it would be of interest to know 
already whether $\sigma(P_\mathbb{N}) < H(1/3)^{-1}$. As we will 
discuss after our applications, bold conjectures in Strassen's 
trilinear theory---Strassen's 
so-called {\em asymptotic rank conjecture}~\cite[Conjecture~5.3]{Strassen1994}
(see also \cite[Problem~15.5]{BurgisserCS1997},
\cite[Conjecture~1.4]{ConnerGLVW2021}, and
\cite[Section 13, p.~122]{WigdersonZ2023})
in particular---imply an affirmative answer and in fact the 
conclusion $\sigma(P_\mathbb{N})=1$.

\medskip
\noindent
{\em Remark.} 
Much as in the study of fast matrix multiplication and with Strassen's 
seminal breakthrough~\cite{Strassen1969} of $\Ra(\mathrm{MM}_2)\leq 7$,
which in the language of exponents translates to $\omega\leq\log_2 7$, 
to obtain nontrivial upper bounds in 
Theorem~\ref{thm:uniform-circuits-for-balanced-tripartitioning}
and in our applications, one needs to only show sufficiently 
low tensor rank for an individual constant-size tensor $P_d$ for some 
constant $d$. This will be immediate from the explicit-decomposition version
of our main Kronecker scaling theorem, Theorem~\ref{thm:kronecker-scaling}, 
in what follows. Here we have, however, chosen to present the introductory 
exposition from the perspective of exponents.

\subsection{Applications}

We are now ready for our applications in arithmetic circuits and algorithms.

\medskip
\noindent
{\em Uniform arithmetic circuits for the permanent.}
The {\em permanent} of a square matrix $A\in\F^{n\times n}$ is 
$\perm A=\sum_{\pi\in S_n} \prod_{i\in[n]} A_{i,\pi(i)}$, where 
the summation is over all permutations $\pi$ of $[n]$. The best general 
algorithm known is due to Ryser~\cite{Ryser1963}, who presented a simple 
inclusion-exclusion formula that can be used to compute the permanent 
with $O(2^nn)$ operations in $\F$. Valiant~\cite{Valiant1979} proved that 
computing the permanent over the integers restricted to matrix entries of 
only zeroes and ones is $\#$P-complete. For this restriction, or more 
generally, matrices with bounded integer values, the fastest known 
algorithm runs in $2^{n-\Omega(\sqrt{n})}$ time, see Li~\cite{Li2023}. 
Bj\"orklund and Williams~\cite{BjorklundW2019} showed that the permanent 
over a finite ring with $r$ elements can be computed 
in $2^{n-\Omega(n/r)}$ time.
Knuth famously asks in {\em The Art of Computer Programming} 
\cite[Volume~2, Exercise~4.6.4.11]{Knuth1998} whether it is possible to 
compute a permanent over the reals with less than $2^n$ arithmetic operations, 
a question that is still open.

As our main application connecting Strassen's theory with Valiant's theory 
and Knuth's question, we show that exponentially smaller arithmetic circuits 
than $2^n$ exist for the permanent under the 
assumption $\sigma(P_\mathbb{N})<H(1/3)^{-1}$. 

\begin{restatable}[Main application; Uniform arithmetic circuits for the permanent]{Thm}{permanent}
\label{thm:main-application-permanent}
For all $\epsilon>0$ there exists an algorithm that given $n$ as input
runs in time  $O\bigl(2^{H(1/3)(\sigma(P_\mathbb{N})+\epsilon)n}\bigr)$
and outputs an arithmetic circuit of size 
$O\bigl(2^{H(1/3)(\sigma(P_\mathbb{N})+\epsilon)n}\bigr)$
for the $n\times n$ permanent.
\end{restatable}

\medskip
\noindent
{\em Uniform arithmetic circuits for the hafnian.}
The {\em hafnian} of a square symmetric matrix $A\in \F^{2n\times 2n}$ is $\haf A=\sum_{p \in P_{2n}^2} \prod_{(i,j)\in p} A_{i,j}$, where $P_{2n}^2$ is the set of all partitions of $[2n]$ into subsets of size $2$. It generalizes the permanent in the sense that it computes the weighted sum over all perfect matchings in an underlying general graph on $2n$ vertices, whereas the permanent computes the weighted sum over all perfect matchings in a bipartite graph. Bj\"orklund~\cite{Bjorklund2012} showed that the hafnian can be computed almost as fast as Ryser's algorithm for the permanent. A simpler algorithm with the same asymptotic running time was given by Cygan and Pilipczuk~\cite{CyganP2015}.

\begin{restatable}[Uniform arithmetic circuits for the hafnian]{Thm}{hafnian}
\label{thm:hafnian}
For all $\epsilon>0$ there exists an algorithm that given $n$ as input
runs in time $O\bigl(2^{H(1/3)(\sigma(P_\mathbb{N})+\epsilon)n}\bigr)$
and outputs an arithmetic circuit of size
$O\bigl(2^{H(1/3)(\sigma(P_\mathbb{N})+\epsilon)n}\bigr)$ for the $2n\times 2n$ hafnian.
\end{restatable}

\medskip
\noindent
{\em Counting set partitions.}
For a set family $\mathcal{F} \subseteq \binom{[n]}{q}$, a \emph{set partition of $\mathcal{F}$} is a subfamily $\mathcal{F}' \subseteq \mathcal{F}$ such that $\mathcal{F}'$ forms a partition of $[n]$.
The number of set partitions can be computed in $|\mathcal{F}|2^n$ time with a folklore dynamic programming algorithm. Using inclusion-exclusion, the problem can also be solved in $2^nn^{O(1)}$ time~\cite{BjorklundHK09} and for constant $q$ there is an algorithm that runs in $2^{n-\Omega(n/q)}$ time~\cite{Koivisto09}. We show exponential improvements independent of $q$, if $\sigma(P_{\mathbb{N}}) < H(1/3)^{-1}$:

\begin{restatable}[Algorithm for counting set partitions]{Thm}{setpartitioning}\label{thm:countsp}
	For all constants $q \in \mathbb{N}$ and $\varepsilon >0$, the number of set partitions of a given family $\mathcal{F} \subseteq \binom{[n]}{q}$ can be computed in $O\bigl(2^{H(1/3)(\sigma(P_\mathbb{N})+\epsilon)n}\bigr)$ time.
\end{restatable}
Our main motivation of this theorem is that it comes tantalizingly close to counting the number of \emph{set covers}: A set cover is a subfamily $\mathcal{F}'\subseteq \mathcal{F}$ such that $\cup_{F \in \mathcal{F}}F=[n]$. Randomized and deterministic algorithms for the \emph{minimization version} of the set cover problem assuming low (asymptotic) rank of $P_d$ were already given in~\cite{BjorklundK2024,BjorklundCHKP2025,Pratt2024}, and one may think that reductions similar to the ones used in these works or~\cite{CyganDLMNOPSW2016} can reduce the problem of counting set covers to the problem of counting set partitions. But this would give a truly interesting breakthrough, since it was shown in~\cite{CyganDLMNOPSW2016} that an $O^*((2-\varepsilon)^n)$ time algorithm that counts all set covers of a family $\mathcal{F} \subseteq \binom{[n]}{q}$ for constant $q$ refutes the Strong Exponential Time Hypothesis (SETH) of Impagliazzo and Paturi~\cite{ImpagliazzoP2001}.

Hence, taking an opportunistic view point, we ask whether Theorem~\ref{thm:countsp} can be extended to counting set covers in the same running time, which would establish a sharp connection between $\sigma(P_{\mathbb{N}})$ and SETH, and show that SETH and the asymptotic rank conjecture are not both true.

The key to the proofs of Theorems~\ref{thm:main-application-permanent}, \ref{thm:hafnian}, and \ref{thm:countsp} is that the values of interest are computed by arithmetic circuits possessing the \emph{skewness} property---that is, at each multiplication gate, at least one of the input polynomials has constant degree.
This allows us to construct circuits via Theorem~\ref{thm:uniform-circuits-for-balanced-tripartitioning} (see Theorem~\ref{thm:faster-inclusion-exclusion} for the construction).

\medskip
\noindent
{\em Multilinear monomial detection.}
The parameterized multilinear monomial detection problem is given an arithmetic circuit representing a multivariate polynomial $P(x)$ over $\mathbb{F}$, decide whether $P(x)$ viewed as a sum of monomials contains a multilinear monomial of degree $k$. Many parameterized  detection problems can be recast as multilinear monomial detection problems;
indeed, some of the best known algorithms for central subgraph detection problems like $k$-path (i.e., finding a simple path of length $k$ in a directed graph), were discovered in this framework~\cite{Koutis08,Williams2009}. 
Originating in the work of Koutis~\cite{Koutis08}, these detection algorithms rely on \emph{weighted counting} in a characteristic two field. In Koutis original paper a polynomial circuit was evaluated over a  group algebra to detect the monomial. Later, Williams~\cite{Williams2009} refined his approach, developing an algorithm with a running time of $O^*(2^k)$.\footnote{The $O^*$ notation suppresses factors polynomial in the input size.}
Koutis and Williams~\cite{KoutisW2016} subsequently showed that very little can be gained by replacing the group algebra for even more complex algebras: there are arithmetic circuits for polynomials encoding the set disjointness problem where the group algebra used by Koutis is provably close to optimal. However, for problems like $k$-path, the arithmetic circuits have the skewness property.
This structural property allows us to bypass the barrier demonstrated by Koutis and Williams~\cite{KoutisW2016}, although our results are conditioned on the tensor rank of balanced tripartitioning.
Among other results, we show the following:

\begin{restatable}{Thm}{kpath}
    For all $\varepsilon >0$, there is a randomized algorithm that, given a directed graph $G$,  decides whether $G$ contains a path of length $k$
    in $O^*(2^{H(1/3)(\sigma(P_\mathbb{N})+\epsilon)k})$ time.
\end{restatable}

Note again that these detection algorithms depend on parity counting, meaning in particular that the earlier connections 
between the asymptotic rank conjecture and combinatorial algorithms in~\cite{BjorklundK2024,Pratt2024,BjorklundCHKP2025} do not apply directly here.

\medskip
\noindent
{\em Hamiltonicity parameterized by Treewidth.}
We also give an application of our method beyond the balanced tripartitioning tensors.
In the \emph{Hamiltonicity} problem one is given an undirected graph, and needs to determine whether it contains a Hamiltonian cycle. It is known that this problem can be solved in $O^*((2+\sqrt{2})^{\pw})$ time when a path decomposition of width $\pw$ is given~\cite{CyganKN2018}, and in $O^*(4^{\tw})$ time when a tree decomposition of width $\tw$ is given~\cite{CyganNPPRW22}.\footnote{The exact definitions are not important and postponed to Section~\ref{sect:treewidth}.}

We define another sequence of \emph{matchings connectivity tensors} $H_{\mathbb{N}}=(H_1,H_2,H_3,\ldots)$ consisting of tensors that indicate whether three matchings join to a single cycle. We show it has the Kronecker scaling property (see Theorem~\ref{thm:matchings-connectivity-tensor-kronecker-scaling}) and give the following algorithmic application:
\begin{restatable}{Thm}{twhamiltonicity}
	For all $\varepsilon >0$, there is a randomized algorithm that takes an $n$-vertex graph $G$ along with a tree decomposition $\mathbb{T}$ of $G$ of treewidth $\tw$ as input, and outputs whether $G$ has a Hamiltonian cycle in time $O^*\left((2+\sqrt{2})^{(\sigma(H_\mathbb{N})+\varepsilon)\tw}\right)$.
\end{restatable}
Let us remark that $\sigma(H_\mathbb{N})=1$, unless there is a three-tensor whose asymptotic rank is larger than its dimensions (and hence a variant of the asymptotic rank conjecture is false).

\subsection{A short discussion on Strassen's asymptotic rank conjecture}

In the language of tensor exponents, 
Strassen conjectured~\cite[Conjecture~5.3]{Strassen1994} that 
the exponent $\sigma(S)=1$ for all tensors $S$ 
that are tight and concise and have shape $d\times d\times d$ for 
some $d=1,2,\ldots$. The conjecture is known to be true for $d=1,2$ but
remains open for $d\geq 3$; already the first open case $d=3$ is of 
considerable interest since a proof for $d=3$ would imply $\omega=2$ by an
application of the Coppersmith--Winograd method~\cite{CoppersmithW1990}
to a particular tensor. The balanced tripartitioning tensors $P_n$ are
known to be both tight and concise, which via the asymptotic scaling identity 
(cf.~Theorem~\ref{thm:asymptotic-scaling})
immediately translates to $\sigma(P_\mathbb{N})=1$. Thus, under Strassen's
asymptotic rank conjecture, Theorem~\ref{thm:main-application-permanent}
yields uniform arithmetic circuits for the permanent that are exponentially 
smaller than Ryser's formula. Also stronger versions of the conjecture
without the tightness and conciseness assumptions appear in the literature
(e.g.~\cite[Problem~15.5]{BurgisserCS1997},
\cite[Conjecture~1.4]{ConnerGLVW2021}, and
\cite[Section 13, p.~122]{WigdersonZ2023}). 

Among the present main evidence towards the conjecture is Strassen's 
result~\cite{Strassen1988} (see also \cite{ChristandlVZ2021,ConnerGLV2022}) 
that $\sigma(S)\leq 2\omega/3=\frac{4}{3}\sigma(\mathrm{MM}_2)$
for any tensor $S$ of shape $d\times d\times d$ for $d=1,2,\ldots$.
It is also known that there are explicit sequences of tensors whose exponent 
conjecture-agnostically captures the worst-case tensor 
exponent $\sigma(d)=\sup_S\sigma(S)$, where $S$ ranges 
over $d\times d\times d$ tensors~\cite{KaskiM2025}.

In motivating the present paper, we prefer a similar, agnostic, view to the 
asymptotic rank conjecture, and would rather like to highlight the 
analysis of the balanced tripartitioning sequence $P_\mathbb{N}$ and 
its exponent $\sigma(P_\mathbb{N})$ as a natural object for further
study. Indeed, each tensor $P_n(x,y,z)$ is invariant under the 
symmetric group $S_{3n}$ acting on $\binom{[3n]}{n}$ and the sets of
indeterminates $x,y,z$ diagonally, suggesting potential for study with 
techniques from representation theory. On the one hand, 
a proof that $\sigma(P_\mathbb{N})<H(1/3)^{-1}$ would via 
Theorem~\ref{thm:main-application-permanent} give a considerable 
advance in the study of the permanent, where progress has yielded
only subexponential speedup since Ryser's 1963 formula. On the other hand,
a proof that $\sigma(P_\mathbb{N})=H(1/3)^{-1}$ would disprove the 
asymptotic rank conjecture. Theorem~\ref{thm:main-application-permanent}
also implies that strong exponential lower bounds for the arithmetic 
complexity of the permanent disprove the asymptotic rank conjecture.

\subsection{Overview of techniques}

Let us now give a brief description of our main theorem
(Theorem~\ref{thm:main-kronecker-scaling}) and its explicit-decomposition
version (Theorem~\ref{thm:kronecker-scaling} in what follows). 
For brevity let $U=[3n]$. The key idea to decompose $P_n$ into a sum of 
restrictions of $P_d^{\otimes s}$ is to assign an {\em intersection type},
or, briefly, {\em type}, $\tau$ to each tripartition $A\cup B\cup C=U$ 
with $|A|=|B|=|C|=n$. Suppose that $n=br$ and $r=gs$ for positive 
integers $b,g,s$. Fix a partition $U=U_1\cup U_2\cup\cdots\cup U_r$ into 
sets $U_i$ with $|U_i|=3b$ for $i\in[r]$. The {\em type} 
$\tau=(\alpha,\beta,\gamma)$ of a tripartition $(A,B,C)$ now consists of 
three-tuples $\alpha,\beta,\gamma\in\{0,1,\ldots,3b\}^r$ with 
$\alpha_i=|A\cap U_i|$,
$\beta_i=|B\cap U_i|$, and 
$\gamma_i=|C\cap U_i|$ for all $i\in[r]$. 
Clearly $\alpha_i+\beta_i+\gamma_i=3b$ for all $i\in[r]$ as well as 
$\sum_{i\in[r]}\alpha_i=n$,
$\sum_{i\in[r]}\beta_i=n$, and
$\sum_{i\in[r]}\gamma_i=n$. Let us write $\tau\in T_b^r$ for the set of
all types. For a type $\tau\in T_b^r$, let us write $\mathcal{P}_U^\tau$ 
for the set of all tripartitions $(A,B,C)$ of $U$ of type $\tau$. 
Since every tripartition has a unique type, we clearly have that the
tensors $P_n^\tau=P_n^\tau(x,y,z)=\sum_{(A,B,C)\in\mathcal{P}_U^\tau}x_Ay_Bz_C$
for $\tau\in T_b^r$ decompose $P_n$ into the sum 
$P_n=\sum_{\tau\in T_b^r}P_n^\tau$. For any $\delta>0$, we can find
$b$ large enough so that $|T_b^r|\leq 2^{\delta n}$, so all we need to do
is to show that each $P_n^\tau$ regardless of the $\tau$ can be obtained as
a restriction of $P_d^{\otimes s}$ with $s\leq(1+\delta)n/d$. We will show
this for $d=b(g+36)$ when $g$ is large enough constant depending on $\delta$. 
Given a type $\tau\in T_b^r$ as input, the key algorithmic idea is to 
efficiently compute a set partition 
$[r]=G_1^\tau\cup G_2^\tau\cup\cdots\cup G_s^\tau$ that for each part 
$j\in[s]$ satisfies $|G_j^\tau|=g$ as well as is {\em balanced} so that 
\[
\biggl|\sum_{i\in G_j^\tau}\alpha_i-bg\biggr|\leq 36b\,,\qquad
\biggl|\sum_{i\in G_j^\tau}\beta_i-bg\biggr|\leq 36b\,,\qquad
\biggl|\sum_{i\in G_j^\tau}\gamma_i-bg\biggr|\leq 36b\,.
\]
This balance property and its efficient computability is crucial in 
embedding $P_n^\tau$ into a restriction of $P_d^{\otimes s}$
for $d=b(g+36)$. We show that the balanced partition 
$[r]=G_1^\tau\cup G_2^\tau\cup\cdots\cup G_s^\tau$ exists and can be 
computed by dynamic programming efficiently enough via a concentration
version (Lemma~\ref{lem:steinitz-concentration}) of the classical Steinitz 
lemma~(Lemma~\ref{lem:steinitz}); the latter shows that a sum of vectors of 
bounded norm in a Euclidean space can be permuted so that all the prefix 
sums adjusted for size closely track the full sum; cf.~\eqref{eq:steinitz}. 
We call this technique {\em Steinitz balancing}. Once each 
decomposition of $P_n^\tau$ as a restriction of $P_d^{\otimes s}$ is available,
our main circuit construction 
(Theorem~\ref{thm:uniform-circuits-for-balanced-tripartitioning}) 
is essentially a consequence of a standard circuit version of 
Yates's algorithm for evaluating Kronecker powers 
(cf.~Lemma~\ref{lem:yates-kron}).

\medskip
\noindent
{\em Remark.}
The proof of Theorem~\ref{thm:uniform-circuits-for-balanced-tripartitioning} only uses some relatively weak closedness properties of the tensor sequence $P_d$ (such as smaller tensors being a restriction of larger tensors in the family) along with constructivity of the decomposition witnessing Kronecker scaling, and hence Kronecker scaling of other tensors can also be consolidated into arithmetic circuits computing the corresponding polynomial. We further exemplify this in Theorem~\ref{thm:hcmain}.

\subsection{Related work}

Kaski and Micha\l{}ek~\cite{KaskiM2025} study tensor sequences that are
universal in the sense that their exponents capture the worst-case 
exponent $\sigma(d)$ for $d\times d\times d$ tensors. Tripartitioning
tensors appear in earlier works of Bj\"orklund and Kaski~\cite{BjorklundK2024}
and Pratt~\cite{Pratt2024}; both works rely on randomization to decide the
existence of a tripartition and essentially do not have the arithmetic and 
counting properties enabled by our present Kronecker scaling decomposition 
and the Steinitz balancing technique. 
Pratt also observes 
the upper bound $\Ra(P_n)\leq 2^{3n-1}$ over any field $\F$ with 
$\operatorname{char}\F\neq 2$. Bj\"orklund, Curticapean, Husfeldt, 
Kaski, and Pratt~\cite{BjorklundCHKP2025} derandomize the randomized 
construction to a deterministic one, and extend the construction to 
unbalanced tripartitioning.

\subsection{Organization of the paper}

The rest of this paper is organized as follows.
Section~\ref{sect:preliminaries} reviews our definitions, notation, and 
background results. Section~\ref{sect:kronecker-scaling} proves our main
results on Kronecker scaling of balanced tripartitioning tensors.
Section~\ref{sect:uniform-circuits-tripartitioning} develops the 
consequent uniform circuit constructions for evaluating balanced 
tripartitioning polynomials. Section~\ref{sect:permanent} proves our 
applications to counting problems, including the permanent in particular.
Section~\ref{sect:detection} shows our applications to parameterized 
problems. Section~\ref{sect:treewidth} concludes the paper by presenting
our application to Hamiltonicity parameterized by treewidth---in particular,
we show that the sequence of matchings connectivity tensors has the
Kronecker scaling property.

\section{Preliminaries}

\label{sect:preliminaries}

This section reviews our key definitions and preliminaries. 

For a nonnegative integer $n$ we write $[n]=\{1,2,\ldots,n\}$.
For a finite set $U$ and a nonnegative integer $k$, 
let us write $\binom{U}{k}$ for the set of all $k$-element subsets of $U$.

For a matrix $A$ over a field $\F$, the entry in the $i$\textsuperscript{th} row and $j$\textsuperscript{th} column is denoted by $A_{i,j}$ or $A[i, j]$. For any sets of row indices $I$ and column indices $J$, the submatrix consisting of the rows in $I$ and the columns in $J$ is denoted by $A[I, J]$. When $I$ includes all rows (or $J$ includes all columns), we write $A[\cdot, J]$ (or $A[I, \cdot]$) as a shorthand.

\subsection{Conventions with tensors}

\label{sect:tensors-conventions}

We work in coordinates and represent tensors as multilinear 
polynomials with the following conventions. 
All of our tensors have order three unless otherwise mentioned.
Let $\F$ be a field and let $U$ be a finite set. 
Let $x,y,z$ be three sets of polynomial indeterminates indexed 
by the subsets of $U$. A {\em tensor} $S\in\F[x,y,z]$ is a multilinear 
polynomial of the form 
\[
S(x,y,z)=\sum_{A,B,C\subseteq U} s_{ABC}x_Ay_Bz_C
\]
with coefficients $s_{ABC}\in\F$. 
We say that $S$ is {\em indexed} by $U$
and that $S$ has {\em shape} $p\times q\times r$ for
$p=|\{A\subseteq U:\text{$s_{ABC}\neq 0$ for some $B,C\subseteq U$}\}|$,
$q=|\{B\subseteq U:\text{$s_{ABC}\neq 0$ for some $A,C\subseteq U$}\}|$, and
$r=|\{C\subseteq U:\text{$s_{ABC}\neq 0$ for some $A,B\subseteq U$}\}|$.

\medskip
\noindent
{\em Kronecker product.}
Let $S\in\F[x,y,z]$ and $T\in\F[x,y,z]$ be tensors indexed by disjoint 
finite sets $U$ and $V$, respectively. The {\em Kronecker product} 
tensor $S\otimes T\in\F[x,y,z]$ is defined by
\[
(S\otimes T)(x,y,z)=
\sum_{A,B,C\subseteq U}\sum_{D,E,F\subseteq V}
s_{ABC}t_{DEF}x_{A\cup D}y_{B\cup E}z_{C\cup F}\,.
\]
In particular, $S\otimes T$ is indexed by $U\cup V$.
For a tensor $S$ and an integer $p$, we write $S^{\otimes p}$ for 
the Kronecker product of $p$ copies of $S$ on pairwise disjoint index sets.
We say that $S^{\otimes p}$ is the $p$\textsuperscript{th} Kronecker power
of $S$.

\medskip
\noindent
{\em Balanced tripartitioning tensors.}
Let $U$ be a set with $3q$ elements for a positive integer $q$. 
The {\em balanced tripartitioning} tensor $P_q[U]\in\F[x,y,z]$ 
is defined by
\[
P_q[U](x,y,z)=\sum_{\substack{A,B,C\in\binom{U}{q}\\A\cup B\cup C=U}}x_Ay_Bz_C\,.
\]

\medskip
\noindent
{\em Tensor rank and asymptotic tensor rank.}
For a tensor $S\in\F[x,y,z]$, the {\em tensor rank} $\Ra(S)$ of $S$ is the 
least nonnegative integer $r$ such that there exist linear polynomials 
$u_i(x)\in\F[x]$, $v_i(y)\in\F[y]$, $w_i(z)\in\F[z]$ for $i\in[r]$ 
with $S(x,y,z)=\sum_{i\in[r]} u_i(x)v_i(y)w_i(z)$.
The {\em asymptotic rank}~\cite{Gartenberg1985} 
of $S$ is $\AR(S)=\lim_{q\rightarrow\infty}\Ra(S^{\otimes q})^{1/q}$,
where the limit exists by Fekete's lemma (see e.g.~\cite{WigdersonZ2023}).
Assuming that $S$ has shape $d\times d\times d$, the asymptotic rank 
$\AR(S)$ and the exponent $\sigma(S)$ satisfy $\AR(S)=d^{\sigma(S)}$.

\subsection{Arithmetic circuits}

Let $x$ be a set of indeterminates and let $\mathbb{F}$ be a field.  
An \emph{arithmetic circuit} over $\mathbb{F}$ (with variables in $x$) is a directed acyclic graph (DAG) defined as follows.  
The indegree-zero nodes of the graph are labeled either by a variable from $x$ or by a constant in $\mathbb{F}$.  
Each internal node $v$ is labeled by either $+$ or $\times$, and it has one or more children nodes computing polynomials $P_1, P_2, \dots, P_r$, with arcs leading from these children into $v$.  
The node $v$ computes $P_1 + P_2 + \cdots + P_r$ (in the case of $+$) or $P_1 \cdot P_2 \cdots P_r$ (in the case of $\times$).
Finally, one or more designated nodes with outdegree zero serve as the output(s) of the circuit.  
The \emph{size} of a circuit is the number of arcs it contains.\footnote{Although the standard measure of an arithmetic circuit's size is the number of gates, we will use the number of arcs instead since it directly corresponds to the number of arithmetic operations required to evaluate the circuit.}

We say that an arithmetic circuit is \emph{homogeneous} if the polynomial computed at every internal node is homogeneous.  
It is possible to transform a nonhomogeneous arithmetic circuit into a homogeneous one.

\begin{Lem}[Homogenization (see, e.g., B\"urgisser~{\cite[Lemma 2.14]{Burgisser2000}})] \label{lem:homogenization}
Any arithmetic circuit of size $s$ computing polynomials of degree at most $d$ can be converted into a homogeneous circuit of size $O(ds)$.
\end{Lem}

A circuit is called \emph{skew} if every multiplication gate has exactly two children and one of these children is an input gate.  
More generally, for a constant $q\in\mathbb{N}$, we say that a circuit is \emph{$q$-skew} if every multiplication gate has exactly two children and at least one of these is computed by a subcircuit that produces a polynomial of degree at most $q$.
Note that the homogenization described in Lemma~\ref{lem:homogenization} preserves the $q$-skew property.

\subsection{Steinitz's lemma}

The following sharp version of Steinitz's~\cite{Steinitz1913} lemma was 
proved by Grinberg and Sevast\cprime janov~\cite{GrinbergS1980}.

\begin{Lem}[Steinitz~\cite{Steinitz1913}; Grinberg and Sevast\cprime janov~{\cite[Theorem~1]{GrinbergS1980}}]
\label{lem:steinitz}
Let an arbitrary norm be given in $\R^d$ and let $u_1,u_2,\ldots,u_r\in\R^d$
with $\|u_i\|\leq 1$ for $i\in [r]$. Then, there exists a permutation
$\pi:[r]\rightarrow[r]$ such that for all $k\in[r]$ we have
\begin{equation}
\label{eq:steinitz}
\biggl\|\sum_{i\in [k]} u_{\pi(i)}-\frac{k-d}{r}\sum_{i\in[r]} u_i\biggr\|\leq d\,.
\end{equation}
\end{Lem}

\medskip
We observe that a permutation $\pi$ that minimizes the maximum of the 
left-hand side of~\eqref{eq:steinitz} over all $k\in[r]$ with respect 
to the infinity
(maximum absolute value coordinate) norm can be found in polynomial time
in $r$ by dynamic programming when there are only $O(1)$ distinct vectors 
among the given $u_1,u_2,\ldots,u_r$ vectors and each vector 
has $O(1)$-bit rational coordinates with $d=O(1)$. This will be the case
in our applications in what follows. 
Indeed, with only $C$ distinct vectors among the
collection of $r$ vectors in the input, 
there are at most $\binom{m+C-1}{C-1}$ distinct subcollections of 
$m\leq r$ vectors obtainable from the input.
We can tabulate for each subcollection of size $1\leq m\leq r$ and each selection of
the $m$\textsuperscript{th} summand in the subcollection 
the optimum min-max value, with the maximum taken over $k\in [m]$.
By tracing the table back one last summand at a time we find an optimum permutation.

\section{Kronecker scaling for balanced tripartitioning tensors}

\label{sect:kronecker-scaling}

This section proves our main finite Kronecker scaling theorem for balanced 
tripartitioning tensors, Theorem~\ref{thm:kronecker-scaling}, 
as well as an asymptotic corollary, 
Theorem~\ref{thm:asymptotic-scaling}.

\subsection{Steinitz concentration}

We start with a simple corollary enabled by Lemma~\ref{lem:steinitz} which 
shows that one can partition a sum with bounded summands into parts such 
that the average of each part concentrates around the global average. 

\begin{Lem}[Steinitz concentration]
\label{lem:steinitz-concentration}
Let an arbitrary norm be given in $\R^d$ and let $v_1,v_2,\ldots,v_r\in\R^d$
with $\|v_i\|\leq 1$ for $i\in [r]$. Let $g_1,g_2,\ldots,g_s$ be positive
integers with $g_1+g_2+\ldots+g_s=r$. Then, there exists a set partition
$G_1\cup G_2\cup\cdots\cup G_s=[r]$ such that 
for all $j\in [s]$ we have $|G_j|=g_j$ and
\[
\biggl\|\frac{1}{g_j}\sum_{i\in G_j} v_i-\frac{1}{r}\sum_{i\in [r]} v_i\biggr\|\leq \frac{4d}{g_j}\,.
\]
\end{Lem}

\begin{Proof}
For $i\in [r]$, let $u_i=\frac{1}{2}v_i-\frac{1}{2r}\sum_{\ell\in[r]}v_\ell$.
Observe by the triangle inequality that $\|u_i\|\leq 1$. 
Also, $\sum_{i\in[r]}u_i=0$. For a nonempty subset $S\subseteq [r]$, 
define $u_S=\sum_{i\in S}u_i$. 
Let $\pi$ be the permutation from Lemma~\ref{lem:steinitz}. 
For each $j\in[s]$, define
\begin{equation}
\label{eq:part-set}
G_j=\{\pi(g_1+g_2+\ldots+g_{j-1}+1),\pi(g_1+g_2+\ldots+g_{j-1}+2),\ldots,\pi(g_1+g_2+\ldots+g_j)\}\,.
\end{equation}
For all $j\in[s]$ we have from~\eqref{eq:steinitz} and \eqref{eq:part-set} that
\begin{equation}
\label{eq:part-steinitz-sum}
\|u_{G_1}+u_{G_2}+\ldots+u_{G_j}\|\leq d\,.
\end{equation}
By the triangle inequality and \eqref{eq:part-steinitz-sum} thus
\[
\biggl\|\frac{1}{2}\sum_{i\in G_j}v_i-\frac{g_j}{2r}\sum_{i\in[r]}v_i\biggr\|=
\|u_{G_j}\|\leq
\|u_{G_1}+u_{G_2}+\ldots+u_{G_{j-1}}\|+\|u_{G_1}+u_{G_2}+\ldots+u_{G_j}\|\leq
2d\,.
\]
\end{Proof}

\medskip
\noindent
{\em Remark.} The partition $G_1,G_2,\ldots,G_s$ 
in Lemma~\ref{lem:steinitz-concentration} is constructible in time
polynomial in $r$ in our applications in what follows; cf. the
paragraph following Lemma~\ref{lem:steinitz}.

\subsection{Kronecker scaling by Steinitz balancing}

We are now ready for our main theorem that establishes the Kronecker
scaling property for balanced tripartitioning tensors. 
Let $b,g,s$ be positive integers and let $q=br$ and $r=gs$. 
Let $U$ be a $3q$-element set. We show how to partition the tensor 
$P_q=P_q[U]$ into disjoint components such that
each component is a restriction of the $s$\textsuperscript{th} 
Kronecker power of $P_{b(g+36)}$. Crucially, we rely on the 
Steinitz concentration lemma (Lemma~\ref{lem:steinitz}) to construct 
the partition into balanced sets in each component---we call this 
technique {\em Steinitz balancing}. 

Partition the set $U$ arbitrarily into $r$ 
sets $U_1,U_2,\ldots,U_r$ of size $3b$ each.
Let $\alpha,\beta,\gamma\in\{0,1,\ldots,3b\}^r$ with
\begin{equation}
\label{eq:q-sum}
\sum_{i\in [r]}\alpha_i=\sum_{i\in [r]}\beta_i=\sum_{i\in [r]}\gamma_i=q
\end{equation}
and 
\begin{equation}
\label{eq:3b-sum}
\alpha_i+\beta_i+\gamma_i=3b 
\end{equation}
for all $i\in[r]$.
We say that the three-tuple $\tau=(\alpha,\beta,\gamma)$ is 
an {\em intersection type}, or {\em type} for short. 
Let us write $T_b^r$ for the set of all intersection types.
Each balanced tripartition $A\cup B\cup C=U$ with 
$A,B,C\in\binom{U}{q}$
now defines a unique type $\tau=(\alpha,\beta,\gamma)$ by 
$\alpha_i=|A\cap U_i|$,
$\beta_i=|B\cap U_i|$, and
$\gamma_i=|C\cap U_i|$ for all $i\in[r]$.
The types $\tau$ will index the disjoint components in our 
decomposition of $P_q[U]$. 

We now proceed with Steinitz balancing.
Fix a type $\tau=(\alpha,\beta,\gamma)\in T_b^r$.
In the Steinitz concentration lemma (Lemma~\ref{lem:steinitz-concentration}), 
take $d=3$, the infinity (maximum absolute value coordinate) norm, $g_1=g_2=\cdots=g_s=g$, and $v_i=\frac{1}{3b}(\alpha_i,\beta_i,\gamma_i)$ for all $i\in[r]$ 
and use \eqref{eq:q-sum} to obtain a set partition
$G_1^\tau\cup G_2^\tau\cup\cdots\cup G_s^\tau=[r]$ such that
for all $j\in[s]$ we have 
\begin{equation}
\label{eq:part-balance}
\biggl|\sum_{i\in G_j^\tau}\alpha_i-bg\biggr|\leq 36b\,,\qquad
\biggl|\sum_{i\in G_j^\tau}\beta_i-bg\biggr|\leq 36b\,,\qquad
\biggl|\sum_{i\in G_j^\tau}\gamma_i-bg\biggr|\leq 36b\,.
\end{equation}
For each $j\in[s]$, introduce a $108b$-element set $V_j$ 
and observe from \eqref{eq:3b-sum} and \eqref{eq:part-balance}
that we can fix an arbitrary set partition 
$V_j^\alpha\cup V_j^\beta\cup V_j^\gamma=V_j$ with 
\begin{equation}
\label{eq:pad-size}
|V_j^\alpha|=bg+36b-\sum_{i\in G_j^\tau}\alpha_i\,,\qquad
|V_j^\beta|=bg+36b-\sum_{i\in G_j^\tau}\beta_i\,,\qquad
|V_j^\gamma|=bg+36b-\sum_{i\in G_j^\tau}\gamma_i\,.
\end{equation}
We assume that the sets $U,V_1,V_2,\ldots,V_s$ are pairwise disjoint. 
For each $j\in [s]$, define $\bar U_j^\tau=(\cup_{i\in G_j^\tau}U_i)\cup V_j$.
Define 
$\bar U=\cup_{j\in [s]}\bar U_j^\tau=U\cup V_1\cup V_2\cup\cdots\cup V_s$.
For all $\bar A,\bar B,\bar C\subseteq\bar U$, define the three 
restrictions
\begin{equation}
\label{eq:bar-substitution}
\begin{split}
\bar x^\tau_{\bar A}&=
\begin{cases}
x_{\bar A\cap U}
& \text{if $|\bar A\cap U_i|=\alpha_i$ for all $i\in [r]$
        and $\bar A\cap V_j=V_j^\alpha$ for all $j\in [s]$;}\\
0 & \text{otherwise},
\end{cases}\\
\bar y^\tau_{\bar B}&=
\begin{cases}
y_{\bar B\cap U}
& \text{if $|\bar B\cap U_i|=\beta_i$ for all $i\in [r]$
        and $\bar B\cap V_j=V_j^\beta$ for all $j\in [s]$;}\\
0 & \text{otherwise},
\end{cases}\\
\bar z^\tau_{\bar C}&=
\begin{cases}
z_{\bar C\cap U}
& \text{if $|\bar C\cap U_i|=\gamma_i$ for all $i\in [r]$
        and $\bar C\cap V_j=V_j^\gamma$ for all $j\in [s]$;}\\
0 & \text{otherwise}.
\end{cases}
\end{split}
\end{equation}

We are now ready for the main Kronecker scaling theorem.

\begin{Thm}[Kronecker scaling for balanced tripartitioning tensors]
\label{thm:kronecker-scaling}
For all positive integers $b,g,s$ and $3bgs$-element sets $U$
we have the polynomial identity
\begin{equation}
\label{eq:kronecker-scaling}
P_{bgs}[U](x,y,z)=
\sum_{\tau\in T_b^{gs}}
\bigl(\bigotimes_{j\in [s]} P_{b(g+36)}[\bar U_j^\tau]\bigr)(\bar x^\tau,\bar y^\tau,\bar z^\tau)\,.
\end{equation}
\end{Thm}
\begin{Proof}
Let $\bar A,\bar B,\bar C\subseteq\bar U$ be arbitrary
and let $\tau\in T_b^{gs}$ be an arbitrary type. 
By definitions of the Kronecker product and balanced tripartitioning 
tensors, we observe that the coefficient
of the monomial $\bar x_{\bar A}\bar y_{\bar B}\bar z_{\bar C}$ in
the polynomial
$\bigl(\bigotimes_{j\in [s]} P_{b(g+36)}[\bar U_j^\tau]\bigr)(\bar x,\bar y,\bar z)$ is $1$ if and only if $(\bar A\cap \bar U_j^\tau,\bar B\cap \bar U_j^\tau,\bar C\cap \bar U_j^\tau)$ is a balanced tripartition of $\bar U_j^\tau$ 
consisting of sets of size $b(g+36)$ for all $j\in [s]$; otherwise the 
coefficient is $0$.
Writing $A=\bar A\cap U$, $B=\bar B\cap U$, and $C=\bar C\cap U$, 
we observe from \eqref{eq:bar-substitution} that 
$\bar x_{\bar A}^\tau\bar y_{\bar B}^\tau\bar z_{\bar C}^\tau=x_Ax_Bx_C$
holds if and only if $(A,B,C)$ has intersection type $\tau$ and 
$(\bar A\cap V_j,\bar B\cap V_j,\bar C\cap V_j)=(V_j^\alpha,V_j^\beta,V_j^\beta)$ for all $j\in[s]$; 
otherwise $\bar x_{\bar A}^\tau\bar y_{\bar B}^\tau\bar z_{\bar C}^\tau=0$.
From \eqref{eq:pad-size} it thus
follows that $\bar x_{\bar A}^\tau\bar y_{\bar B}^\tau\bar z_{\bar C}^\tau=x_Ax_Bx_C$
if and only if $(A,B,C)$ is a balanced tripartition of $U$ with 
intersection type $\tau$; otherwise 
$\bar x_{\bar A}^\tau\bar y_{\bar B}^\tau\bar z_{\bar C}^\tau=0$.
Moreover, 
when $\bar x_{\bar A}^\tau\bar y_{\bar B}^\tau\bar z_{\bar C}^\tau=x_Ax_Bx_C$,
the balanced tripartition $(A,B,C)$ uniquely determines the 
balanced tripartition $(\bar A,\bar B,\bar C)$ by 
$\bar A=A\cup\bigcup_{j\in[s]}V_j^\alpha$,
$\bar B=B\cup\bigcup_{j\in[s]}V_j^\beta$, and
$\bar C=C\cup\bigcup_{j\in[s]}V_j^\gamma$.
The identity~\eqref{eq:kronecker-scaling} now follows since 
every tripartition
$(A,B,C)$ of $U$ with $A,B,C\in\binom{U}{bgs}$ has a unique intersection
type $\tau\in T_b^{gs}$ and we sum over all such types. 
\end{Proof}

Theorem~\ref{thm:kronecker-scaling} enables an immediate proof of 
Theorem~\ref{thm:main-kronecker-scaling} that we supply now for completeness.

\kroneckerscaling*
\begin{Proof}
Fix an arbitrary $\delta>0$. As suggested by 
Theorem~\ref{thm:kronecker-scaling}, let us consider the tensors $P_n$ with $n$ 
an integer of the form $n=bgs$ for positive integers $b,g,s$, 
where $b$ and $g$ will be large enough constants to be selected in what 
follows, and will $s$ grow without bound; that is, $n$ will belong to 
the arithmetic progression $\{bgs:s=1,2,\ldots\}$. 
Assume that (i) $b$ is large enough so that 
$|T_b^{gs}|\leq (3b+1)^{3gs}=((3b+1)^{3/b})^n\leq 2^{\delta n}$; 
this ensures that the sum in \eqref{eq:kronecker-scaling} ranges
over at most $2^{\delta n}$ tensors. 
Similarly, assume that (ii) $g$ is large enough so that $g+36\leq (1+\delta)g$;
this ensures, taking $d=b(g+36)$ and considering the tensor $P_d$, 
by \eqref{eq:kronecker-scaling} that $P_n$ is a sum of restrictions of 
$P_d^{\otimes s}$ with $s=n/(bg)\leq (1+\delta)n/(b(g+36))=(1+\delta)n/d$.
We also observe that by increasing $b$ and $g$ as necessary
we obtain infinitely many such $d$ that meet the assumptions (i) and (ii).
\end{Proof}

\subsection{Asymptotic scaling}

Let us now derive an asymptotic consequence of 
Theorem~\ref{thm:kronecker-scaling}.
We abbreviate $P_n$ for the tensor $P_n[U]$ with $U=[3n]$. 
We also recall from Section~\ref{sect:introduction} 
that we write $\sigma(P_\mathbb{N})$ 
for the exponent \eqref{eq:exponent} of balanced 
tripartitioning~\eqref{eq:balanced-tripartitioning}.

\begin{Thm}[Asymptotic scaling for balanced tripartitioning tensors]
\label{thm:asymptotic-scaling}
We have
\begin{equation}
\label{eq:asymptotic-scaling}
\sigma(P_\mathbb{N})=
\inf\,\biggl\{\sigma>0:\Ra(P_n)\leq\binom{3n}{n}^{\sigma+o(1)}\biggr\}=\inf\,\biggl\{\sigma>0:\AR(P_n)\leq\binom{3n}{n}^{\sigma+o(1)}\biggr\}\,.
\end{equation}
\end{Thm}
\begin{Proof}
The leftmost identity in \eqref{eq:asymptotic-scaling} holds by definition.
By properties of tensor rank and asymptotic rank, it is immediate 
that $\sigma=2$ belongs to both sets in \eqref{eq:asymptotic-scaling}, 
so both sets are nonempty and bounded from below. 
Let $\sigma_0$ be an arbitrary element 
of $\{\sigma>0:\Ra(P_n)\leq\binom{3n}{n}^{\sigma+o(1)}\}$. 
Fix an arbitrary $\epsilon>0$ and observe that 
$\Ra(P_n)\leq \binom{3n}{n}^{\sigma_0+\epsilon}$ holds for
all large enough~$n$. Since tensor rank is an upper bound for asymptotic
rank, $\AR(P_n)\leq \Ra(P_n)$ in particular, we conclude that 
$\sigma_0+\epsilon$ is in 
$\{\sigma>0:\AR(P_n)\leq\binom{3n}{n}^{\sigma+o(1)}\}$.

Let $\sigma_0$ be an arbitrary element of 
$\{\sigma>0:\AR(P_n)\leq\binom{3n}{n}^{\sigma+o(1)}\}$.
Fix an arbitrary $\epsilon>0$ and an arbitrary $\delta>0$.
Observe that $\AR(P_n)\leq \binom{3n}{n}^{\sigma_0+\epsilon}$ holds for
all large enough~$n$. Assume such an $n$ has been fixed.
By definition of asymptotic rank, 
$\Ra\bigl((P_n)^{\otimes p}\bigr)\leq \binom{3n}{n}^{(\sigma_0+\epsilon+\delta)p}$ holds for all large enough $p$. 
From \eqref{eq:kronecker-scaling} as well as by subadditivity of 
tensor rank for all positive integers $b,g,s$ we have 
\[
\Ra(P^{3bgs}_{bgs})\leq |T^{gs}_b|\cdot \Ra\bigl((P^{3b(g+36)}_{b(g+36)})^{\otimes s}\bigr)\,.
\]
Assuming that $b(g+36)$ and $s$ are large enough, 
and using Stirling's formula (see e.g.~Robbins~\cite{Robbins1955})
to bound the binomial coefficient from above, we thus have
\[
\Ra(P^{3bgs}_{bgs})\leq 
(3b+1)^{3gs}\cdot 
\binom{3b(g+36)}{b(g+36)}^{(\sigma_0+\epsilon+\delta)s}
\leq 
(3b+1)^{3gs}\cdot 
2^{H(1/3)\cdot 3b(g+36)(\sigma_0+\epsilon+\delta)s}
\,,
\]
where $H(\lambda)=-\lambda\log_2 \lambda-(1-\lambda)\log_2(1-\lambda)$ 
is the binary entropy function.
Writing $m=bgs$, we thus have
\[
\Ra(P^{3m}_{m})\leq 
2^{3\frac{\log_2(3b+1)}{b}m}
\cdot
2^{H(1/3)\cdot 3(1+\frac{36}{g})(\sigma_0+\epsilon+\delta)m}\,.
\]
Assuming that $b$ and $g$ are large enough constants, and using 
Stirling's formula to bound the binomial coefficient from below, 
for all large enough integer multiples $m$ of $bg$ we conclude that 
\[
\Ra(P^{3m}_{m})\leq 
2^{H(1/3)\cdot 3(\sigma_0+2\epsilon+2\delta)m}
\leq \binom{3m}{m}^{\sigma_0+3\epsilon+3\delta}\,.
\]
The assumption that $m$ is a multiple of the constant $bg$ can be lifted by a
construction analogous to the ``padding and restriction'' construction
$A\mapsto\bar A$,
$B\mapsto\bar B$,
$C\mapsto\bar C$
given in the proof of Theorem~\ref{thm:kronecker-scaling}; 
we omit the details and conclude that for all large enough $m$ we have
\[
\Ra(P^{3m}_{m})\leq\binom{3m}{m}^{\sigma_0+4\epsilon+4\delta}\,.
\]
We conclude that $\sigma_0+4\epsilon+4\delta$ is in 
$\{\sigma>0:\Ra(P_n)\leq\binom{3n}{n}^{\sigma+o(1)}\}$.
\end{Proof}

\section{Uniform circuits for the balanced tripartitioning polynomial}

\label{sect:uniform-circuits-tripartitioning}

This section gives our main arithmetic circuit construction relying
on Theorem~\ref{thm:kronecker-scaling} and proves 
Theorem~\ref{thm:uniform-circuits-for-balanced-tripartitioning}.
We start with short and well-known preliminaries on evaluating 
a Kronecker power of a tensor using Yates's algorithm~\cite{Yates1937}.

\subsection{Yates's algorithm and circuits for evaluating Kronecker powers}

The following lemma is a standard application of 
Yates's algorithm~\cite{Yates1937} viewed as a circuit, 
and holds also when rank is replaced with asymptotic rank.
For completeness, we give a concise proof but stress that the result
is well known.

\begin{Lem}[Evaluation of Kronecker powers]
\label{lem:yates-kron}
Let $T$ be a tensor of shape $d\times d\times d$ and rank at most $r$ over a 
field\/ $\F$ for some constants $r\geq d$. Then, for all $\epsilon>0$ 
and all positive integers $s$ there exists an $\F$-arithmetic circuit 
of size $O(r^{(1+\epsilon)s})$ 
and depth $O(s)$ constructible in time $O(r^{(1+\epsilon)s})$
that given values in $\F$ to the variables $x,y,z$ as input outputs 
the value of the Kronecker power polynomial $T^{\otimes s}(x,y,z)$. 
\end{Lem}

\begin{Proof}{}
Indexing the polynomial indeterminates of $T(x,y,z)\in\F[x,y,z]$ 
by $[d]$ rather than sets, and recalling 
Section~\ref{sect:tensors-conventions}, the assumption $\Ra(T)\leq r$ 
directly implies there exist 
matrices $U,V,W\in\F^{d\times r}$ satisfying the polynomial identity
\[
T(x,y,z)=\sum_{\ell\in[r]}
\biggl(\sum_{i\in [d]}U_{i,\ell}x_i\biggr)
\biggl(\sum_{j\in [d]}V_{j,\ell}y_j\biggr)
\biggl(\sum_{k\in [d]}W_{k,\ell}z_k\biggr)
\,.
\]
Accordingly, the Kronecker power $T^{\otimes s}(x,y,z)\in\F[x,y,z]$ 
satisfies the identity
\begin{equation}
\label{eq:kronpow}
\begin{split}
T^{\otimes s}(x,y,z)=\sum_{\ell\in[r]^s}&
\biggl(\sum_{i\in [d]^s}U_{i_1,\ell_1}U_{i_2,\ell_2}\cdots U_{i_s,\ell_s}x_i\biggr)\\
&\biggl(\sum_{j\in [d]^s}V_{j_1,\ell_1}V_{j_2,\ell_2}\cdots V_{j_s,\ell_s}y_j\biggr)\\
&\biggl(\sum_{k\in [d]^s}W_{k_1,\ell_1}W_{k_2,\ell_2}\cdots W_{k_s,\ell_s}z_k\biggr)
\,,
\end{split}
\end{equation}
where we write $[d]^s$ and $[r]^s$ for the Cartesian product of $s$ copies of
$[d]$ and $[r]$, respectively. The identity \eqref{eq:kronpow} also gives
an immediate formula for computing $T^{\otimes s}(x,y,z)$ from the 
inputs $x,y,z$; however, the formula does not meet the size requirement. 
To meet the size requirement, it suffices to design an arithmetic circuit 
of size $O(r^{(1+\epsilon)s})$ that given $x_i$ for each $i\in [d]^s$ as 
input, outputs the values 
$\hat x_\ell=\sum_{i\in [d]^s}U_{i_1,\ell_1}U_{i_2,\ell_2}\cdots U_{i_s,\ell_s}x_i$ for each $\ell\in [r]^s$.
The circuit, essentially Yates's algorithm~\cite{Yates1937}, consists 
of $s+1$ {\em layers}, with layer $u$ taking input from layer
$u-1$ for $u=1,2,\ldots,s$. Let us denote the essential gates in layer $u$ by 
$g^{[u]}_{\ell_1,\ell_2,\ldots,\ell_u,i_{u+1},i_{u+2},\ldots,i_s}$ 
with $\ell_1,\ell_2,\ldots,\ell_u\in[r]$ and $i_{u+1},i_{u+2},\ldots,i_s\in[d]$.
The input is at layer $0$ with $g^{[0]}_i=x_i$ for all $i\in [d]^s$, and the
output is given at layer $s$ with $g^{[\ell]}_\ell=\hat x_\ell$ 
for all $\ell\in[r]^s$. The circuit in layer $u=1,2,\ldots,s$ is defined by
for all $\ell_1,\ell_2,\ldots,\ell_u\in[r]$ and 
$i_{u+1},i_{u+2},\ldots,i_s\in[d]$ by the rule
\[
g^{[u]}_{\ell_1,\ell_2,\ldots,\ell_u,i_{u+1},i_{u+2},\ldots,i_s}
\leftarrow\sum_{i_u\in[s]}
U_{i_u,\ell_u}
g^{[u-1]}_{\ell_1,\ell_2,\ldots,\ell_{u-1},i_{u},i_{u+1},\ldots,i_s}\,.
\]
We omit the proof of correctness by 
induction on $u$ as well as the circuit size analysis using the sum of 
a geometric
series and $r\geq d$. Here we only described the subcircuit for computing
the parenthesized expressions involving $U$ and $x$ in \eqref{eq:kronpow};
the circuits involving $V$ and $y$ as well as $W$ and $z$ are identical.
This completes the circuit design.
\end{Proof}

\subsection{The balanced tripartitioning polynomial}

This section proves our main evaluation theorem, 
Theorem~\ref{thm:uniform-circuits-for-balanced-tripartitioning}, 
for the balanced tripartitioning polynomial $P_n(x,y,z)$ 
using Theorem~\ref{thm:kronecker-scaling} and Lemma~\ref{lem:yates-kron}.

In the language of exponents, we will also prove the following corollary 
based on the balanced tripartitioning exponent $\sigma(P_\mathbb{N})$; 
also recall Theorem~\ref{thm:asymptotic-scaling}.

\begin{Thm}[Uniform circuits for balanced tripartitioning polynomials; exponent version] \label{thm:tripartitioning}
Let $\F$ be a field. 
For all $\epsilon>0$ and all positive integers $n$ there exists an
$\F$-arithmetic circuit of size 
$O(\binom{3n}{n}^{\sigma(P_\mathbb{N})+\epsilon})$
constructible in time 
$O(\binom{3n}{n}^{\sigma(P_\mathbb{N})+\epsilon})$
that given values in $\F$ to the variables $x,y,z$ as input
outputs the value of the balanced three-way partitioning polynomial
$P_n(x,y,z)$. 
\end{Thm}

We start with a proof of 
Theorem~\ref{thm:uniform-circuits-for-balanced-tripartitioning}.

\tripartitioncircuit*
\begin{Proof}
Let $\Lambda\geq 1$ be a constant such that $\Ra(P_d)\leq\Lambda^d$ for all
large enough $d$. By flattening $P_d$ into a matrix and observing a large
identity submatrix, we have that 
$\Ra(P_d)\geq\binom{3d}{d}$ and thus by Stirling's formula we can 
assume that $\Lambda\geq 2^{3H(1/3)}$, implying that we can take
$r=\lfloor\Lambda^d\rfloor$ in Lemma~\ref{lem:yates-kron}. 
Now select an arbitrary $\Gamma>\Lambda$ and suppose that $n=1,2,\ldots$ is
given as input. Working with the positive integer parameters $b,g,s$ in 
Theorem~\ref{thm:kronecker-scaling}, and assuming that $b,g$ are constants
with $bg\geq 2$ whose values are selected in what follows, select the unique 
$s=1,2,\ldots$ so that $bg(s-1)< n\leq bgs$.
Now, choose the constants $b$ and $g$ to be large enough, as well as 
a constant $\epsilon>0$ that is small enough, so that
\begin{equation}
\label{eq:constant-seln}
(3b+1)^{3/b}\Lambda^{(1+\epsilon)(1+36/g)}<\Gamma\,.
\end{equation}
The circuit construction now proceeds as follows. First, using 
Lemma~\ref{lem:yates-kron}, build a circuit for $P_d^{\otimes s}$ with
$d=b(g+36)$. This construction runs in time $O(\Lambda^{(1+\epsilon)ds})$
and produces a circuit $\bar C$ of similar size with inputs indexed by 
$b(g+36)s$-subsets of $\bar U$ with $|\bar U|=3b(g+36)s$. 
Then, using the construction in the proof of 
Theorem~\ref{thm:kronecker-scaling}, take $|T_b^{gs}|$ copies
of the constructed circuit $\bar C$, with each copy indexed by a unique 
$\tau\in T_b^{gs}$, and restrict/substitute inputs to the circuit $C$ 
as in \eqref{eq:bar-substitution} to inputs indexed by $bgs$-subsets 
of $U$ with $|\bar U|=bgs$; this results in a circuit $C_\tau$. 
Finally, take the sum of the outputs of the circuits $C_\tau$ 
over $\tau\in T_b^{gs}$ to obtain the circuit $C$ that computes
the polynomial $P_{bgs}$. We observe that $C$ has size at most
$O(|T_b^{gs}|\Lambda^{(1+\epsilon)ds})$ and can be constructed in similar 
time; indeed, observe that the restriction/substitution 
\eqref{eq:bar-substitution} can be computed from $\tau$ using the 
partitioning algorithm highlighted in the remark after the Steinitz 
concentration lemma (Lemma~\ref{lem:steinitz-concentration}) as well as
the paragraph after Lemma~\ref{lem:steinitz}. From \eqref{eq:constant-seln}
and the choice of $s$ we now observe that
\[
|T_b^{gs}|\Lambda^{(1+\epsilon)ds}\leq
\bigl((3b+1)^{3/b}\Lambda^{(1+\epsilon)(1+36/g)}\bigr)^{bgs}
<\Gamma^{bg}\Gamma^n\,,
\]
which is $O(\Gamma^n)$ since $b$ and $g$ are constants.
\end{Proof}

We conclude this section with the proof of Theorem~\ref{thm:tripartitioning}.

\begin{Proof}[Proof of Theorem~\ref{thm:tripartitioning}]
Fix an arbitrary $\epsilon>0$. By \eqref{eq:asymptotic-scaling}
for all large enough $d$ it holds that 
$\Ra(P_d)\leq\binom{3d}{d}^{\sigma(P_\mathbb{N})+\epsilon/3}$,
so by Stirling's formula we can take 
$\Lambda=2^{3H(1/3)(\sigma(P_\mathbb{N})+\epsilon/3)}$
and $\Gamma=2^{3H(1/3)(\sigma(P_\mathbb{N})+2\epsilon/3)}>\Lambda$
in Theorem~\ref{thm:uniform-circuits-for-balanced-tripartitioning}
to obtain circuits of size $O(\Gamma^n)$ constructible in similar time.
Since $\Gamma^n\leq \binom{3n}{n}^{\sigma(P_\mathbb{N})+\epsilon}$ 
for all large enough $n$ by Stirling's formula, the present theorem follows.
\end{Proof}

\section{Applications to counting problems}
\label{sect:permanent}

In this section, we present our results for various counting problems.  
We begin with the permanent and then move on to more general results for dynamic programming over subsets implemented by skew circuits.  
Finally, we discuss several applications, including the hafnian and the set partitioning problem.

\subsection{Permanent}
In this subsection, we present a circuit construction for the permanent:

\permanent*
\begin{proof}

    Let $A$ be an $n \times n$ matrix. Recall that the permanent of $A$ is given by
    \[
    \perm A = \sum_{M} w(M),
    \]
    where the sum is over all perfect matchings $M$ in the complete bipartite graph on $[n] \times [n]$, and 
    \[
    w(M) = \prod_{(i,j) \in M} A[i,j].
    \]
    A standard dynamic programming approach computes this sum by building up contributions from partial matchings.

    In our construction, we assume that $n$ is a multiple of three and partition the $n$ rows into three contiguous blocks of size $n/3$. For each block (indexed by $\ell\in [3]$), we construct a set of gates $g^\ell_U$, where $U\in \binom{[n]}{i}$ for $1\le i\le n/3$. The intended meaning of the gate $g^\ell_U$ is to compute the sum of weights corresponding to all partial matchings in the $\ell$-th block that cover exactly the columns in~$U$. In particular, the recursion is defined as follows:
    \begin{enumerate}
        \item For each singleton $U = \{j\}$, the gate $g^\ell_U$ is an input gate corresponding to the entry in the $i$\textsuperscript{th} row and the $j$\textsuperscript{th} column:
        \[
        g^\ell_{\{j\}} = a_{(\ell-1)n/3+1,j}.
        \]
        \item For each $i \in [n/3]$ with $i \ge 2$ and for each $U \in \binom{[n]}{i}$, we construct $i$ multiplication gates. For each $j \in U$, the corresponding multiplication gate computes
        \[
        a_{(\ell-1)n/3+i,j}\cdot g^\ell_{U \setminus \{j\}}\,.
        \]
        Then, the gate $g^\ell_U$ is defined as the sum of these $i$ products:
        \[
        g^\ell_U = \sum_{j\in U} a_{(\ell-1)n/3+i,j}\cdot g^\ell_{U\setminus \{j\}}\,.
        \]
    \end{enumerate}
    By an inductive argument, one can verify that for each block $\ell$, the gate $g^\ell_U$ computes the sum of weights over all partial matchings (restricted to the $\ell$-th block) that cover the columns in $U$.

    Finally, we combine the contributions from the three blocks using Theorem~\ref{thm:tripartitioning}.
    Since every perfect matching in the bipartite graph can be partitioned into three parts (one for each block), the permanent of $A$ is computed by the combined circuit:
    \[
    \perm A = \sum_{\substack{(U_1  U_2, U_3) \text{ is a balanced}\\\text{tripartition of } [n]}} g^1_{U_1}\cdot g^2_{U_2}\cdot g^3_{U_3}.
    \]
    The bottom part of the circuit has size $O(\binom{n}{n/3}n)$, and by Theorem~\ref{thm:tripartitioning}, the top part of the circuit has size $O\bigl(2^{H(1/3)(\sigma(P_\mathbb{N})+\epsilon)n}\bigr)$. Both parts can be constructed in $O\bigl(2^{H(1/3)(\sigma(P_\mathbb{N})+\epsilon)n}\bigr)$ time.
\end{proof}

\subsection{Subset dynamic programming}
In this section, we show how Theorem~\ref{thm:main-application-permanent} can be further generalized to cover dynamic programming over subsets implemented via skew circuits.  
We begin with a standard construction and then present an alternative construction using Theorem~\ref{thm:tripartitioning}.  
We show that Theorem~\ref{thm:tripartitioning} provides a novel and versatile tool for constructing arithmetic circuits for subset dynamic programming.  
Although the underlying proof employs standard techniques, the resulting framework is quite powerful.
Indeed, in Subsection~\ref{subsec:counting-applications}, we will demonstrate several examples to illustrate its applications.

As a warmup, we start with a circuit construction that does not yet use Theorem~\ref{thm:tripartitioning}.

\begin{Lem}[Construction for subset dynamic programming] \label{lem:dp}
Let $x$ be a set of variables indexed by $[n]$ and let $\mathbb{F}$ be a field.
Suppose there exists a polynomial-size 1-skew arithmetic circuit $C$ that computes a polynomial $P(x)$ of degree $n$ over $\mathbb{F}$.
There exists an algorithm that given $C$ as input runs in time $O^*(2^n)$
and outputs an arithmetic circuit of size $O^*(2^n)$ that computes the coefficient of $\prod_{i=1}^n x_i$ in $P(x)$.
\end{Lem}

\begin{proof}
    We replace each internal gate $g$ in $C$ with a collection of $2^n$ gates $g_S$, for every $S \subseteq [n]$, that compute the coefficient of the monomial $\prod_{i\in S}x_i$ in the polynomial computed at $g$. The gate corresponding to $S=[n]$ is designated as the output. The remaining gates are handled in the natural manner. In particular, for a multiplication gate $g=g'\cdot g''$, where by 1-skewness $g'$ has degree at most 1, we compute for each $S\subseteq [n]$:
    \[
    g_S = g'_\emptyset \cdot g''_S + \sum_{i\in S} g'_{\{i\}} \cdot g''_{S\setminus\{i\}},
    \]
    which uses $|S|+1$ multiplication gates. Since each gate has fan-in at most $n+1$, the overall size of the constructed circuit is $O^*(2^n)$.
\end{proof}

We now proceed to a construction that leverages Theorem~\ref{thm:tripartitioning}.  
The key idea is to apply the homogenization procedure (Lemma~\ref{lem:homogenization}), which allows us to effectively partition the circuit into three layers.  
Subsequently, we use Theorem~\ref{thm:tripartitioning} to combine the results from each layer.

\begin{restatable}[Construction for subset dynamic programming via Theorem~\ref{thm:tripartitioning}]{Thm}{subsetdp} \label{thm:faster-inclusion-exclusion}
    Let $x$ be a set of variables indexed by $[n]$ and let $\mathbb{F}$ be a field.
    Suppose there exists a polynomial-size 1-skew arithmetic circuit $C$ that computes a polynomial $P(x)$ of degree $n$ over $\mathbb{F}$.
    For all $\varepsilon > 0$, there exists an algorithm that given $C$ as input runs in time $O\bigl(2^{H(1/3)(\sigma(P_\mathbb{N})+\epsilon)n}\bigr)$
    and outputs an arithmetic circuit of size $O\bigl(2^{H(1/3)(\sigma(P_\mathbb{N})+\epsilon)n}\bigr)$ that computes the coefficient of $\prod_{i=1}^n x_i$ in $P(x)$.
\end{restatable}

\begin{proof}
We assume that $n \ge 9$ (otherwise the coefficient can be computed in constant time) and that $n$ is a multiple of three.
Since we are interested in the coefficient of $\prod_{i=1}^n x_i$, by the homogenization (Lemma~\ref{lem:homogenization}) we may assume that $P(X)$ is homogeneous of degree $n$, and that $C$ is a homogeneous 1-skew circuit computing $P(X)$.

For each $i\in \{0,1,\dots,n\}$, let 
\[
G_i = \{ \text{gates in } C \text{ that compute a polynomial of degree } i \}.
\]
Because $C$ is homogeneous and $1$-skew, the following holds:
\begin{itemize}
    \item For every addition gate in $G_i$, both inputs must lie in $G_i$.
    \item For every multiplication gate in $G_i$, by the $1$-skew property one of the inputs has degree at most $1$. Hence, either one input is from $G_{i-1}$ and the other from $G_1$ (so that their product has degree $i$), or one input is from $G_i$ and the other from $G_0$ (i.e., a constant).
\end{itemize}

We now partition the circuit $C$ into three subcircuits $C_1$, $C_2$, and $C_3$ according to the degree layers:
\begin{enumerate}
    \item $C_1$: Restrict $C$ to the gates in 
    \[
    G_0 \cup G_1 \cup \cdots \cup G_{n/3}.
    \]
    In $C_1$, we designate all gates in $G_{n/3}$ as outputs.
    
    \item $C_2$: Restrict $C$ to the gates in 
    \[
    G_0 \cup G_1 \cup G_{n/3} \cup G_{n/3+1} \cup \cdots \cup G_{2n/3}.
    \]
    In this subcircuit, treat the gates in $G_{n/3}$ as inputs (introducing a new variable set 
    \[
    Y = \{ y_1, \dots, y_s \},
    \]
    in place of the actual polynomial outputs; here we remove the arcs that connect into addition gates) and designate the gates in $G_{2n/3}$ as outputs.
    
    \item $C_3$: Restrict $C$ to the gates in 
    \[
    G_0 \cup G_1 \cup G_{2n/3} \cup G_{2n/3+1} \cup \cdots \cup G_n.
    \]
    Here, treat the gates in $G_{2n/3}$ as inputs (using a new variable set 
    \[
    Z = \{ z_1, \dots, z_t \},
    \]
    distinct from $Y$; again, we remove the arcs that connect into addition gates), and designate the overall output gate of $C$ as the output of $C_3$.
\end{enumerate}

By construction, each output of $C_1$ is a homogeneous polynomial of degree $n/3$.  
Denote these outputs by $f_1(X), \dots, f_s(X)$, which serve as the input variables $Y$ in $C_2$.

Next, we argue that each output of $C_2$ is a linear form in the new variables $Y$.  
Since $C$ is 1-skew and $n \ge 9$, a simple induction on the degree layers shows that, for each $i \in [n/3]$, every gate in $G_{n/3+i}$ computes a polynomial that is linear in the variables from $Y$.  
Thus, for $j \in [t]$, the $j$\textsuperscript{th} output of $C_2$, which serves as the input variable $z_j$ for $C_3$, can be expressed as
\[
\sum_{i=1}^{s} y_i\, g_{i,j}(X),
\]
where each $g_{i,j}(X)$ is a homogeneous polynomial of degree $n/3$.

Similarly, the output of $C_3$ can be expressed as
\[
\sum_{j=1}^{t} z_j\, h_j(X),
\]
where each $h_j(X)$ is a homogeneous polynomial of degree $n/3$.

Note that arithmetic circuits for computing the multilinear parts of the polynomials $f_i(X)$, $g_i(X)$, and $h_j(X)$ can be constructed in $O^*(2^{H(1/3)n})$ time, as shown in the proof of Lemma~\ref{lem:dp}.

To recover the output of the original circuit $C$, we substitute the expressions from $C_2$ into the inputs $z_j$ of $C_3$, followed by a further substitution of the outputs of $C_1$ for the variables $y_i$. This yields an expression
\[
    \sum_{i=1}^{s} \sum_{j=1}^{t} f_i(X) g_{i,j}(X) h_j(X).
\]

Since we are interested only in the coefficient of the multilinear monomial $\prod_{i=1}^n x_i$ in the final output, it suffices to extract the multilinear part of the above expression.
Since we have already constructed arithmetic circuits for computing the multilinear parts of $f_i(X)$, $g_{i,j}(X)$, and $h_j(X)$, and since $C$ is of polynomial size (so that $s$ and $t$ are polynomially bounded), Theorem~\ref{thm:tripartitioning} implies that an arithmetic circuit of size $O\bigl(2^{H(1/3)(\sigma(P_\mathbb{N})+\epsilon)n}\bigr)$ computing the coefficient of $\prod_{i=1}^n x_i$ can be constructed in $O\bigl(2^{H(1/3)(\sigma(P_\mathbb{N})+\epsilon)n}\bigr)$ time.
\end{proof}

\medskip
\noindent
{\em Remark.}
Though Theorem~\ref{thm:faster-inclusion-exclusion} is stated for 1-skew circuits, it can be easily generalized to $q$-skew arithmetic circuits for $q \in O(1)$.

\subsection{Applications.} \label{subsec:counting-applications}
In this subsection we demonstrate three applications of Theorem~\ref{thm:faster-inclusion-exclusion}.

\medskip
\noindent
{\em Permanent.} We start with the permanent, recovering Theorem~\ref{thm:main-application-permanent}.
To that end, it suffices to show that the permanent can be computed by a 1-skew circuit.

\begin{Lem} \label{lem:permanent}
Let $A \in \F^{n \times n}$ and $x = \{ x_1, \dots, x_n \}$. Then, the permanent $\perm A$ can be computed as the coefficient of the monomial $\prod_{i=1}^n x_i$ in a polynomial $P(x)$ that can be computed by a polynomial-size $1$-skew arithmetic circuit.
\end{Lem}

\begin{proof}
Consider the polynomial
\begin{align*}
P(x) = \prod_{j=1}^n \sum_{i=1}^n x_i A[i, j].
\end{align*}
Since it consists of a product of sums, it can be computed by a polynomial-size $1$-skew arithmetic circuit. Expanding the product, we obtain
\begin{align*}
P(x) = \sum_{f \colon [n] \to [n]} \prod_{j=1}^n x_{f(j)} A[f(j), j],
\end{align*}
where $f$ ranges over all mappings from $[n]$ to $[n]$.

Extracting the coefficient of $\prod_{i=1}^n x_i$ corresponds to selecting only the terms where each $x_i$ appears exactly once. This happens precisely when $f$ is a bijection, meaning $f$ is a permutation of~$[n]$. Since the permanent is defined as the sum over all such permutations, we conclude that the coefficient of $\prod_{i=1}^n x_i$ is exactly $\perm A$.
\end{proof}

Theorem~\ref{thm:faster-inclusion-exclusion} combined with Lemma~\ref{lem:permanent} immediately yields an alternative proof of Theorem~\ref{thm:main-application-permanent}.

\begin{figure}
\begin{tikzpicture}[scale=0.7,shorten >=1pt, auto, node distance=1cm, ultra thick]
 
    \tikzstyle{node_stylea} = [circle, draw=black, fill=lightgray, inner sep=1pt, outer sep=-1pt, minimum size=16pt, line width=.5]
    \tikzstyle{node_stylek} = [circle, draw=black, fill=white, inner sep=1pt, outer sep=-1pt, minimum size=16pt, line width=.5]
    \tikzstyle{node_styleh} = [circle, draw=white, fill=white, inner sep=1pt, outer sep=-1pt, minimum size=16pt, line width=.5]
  
    \tikzstyle{t} = [draw=red, line width=1.5]
    \tikzstyle{s} = [draw=black, line width=0.5]
    \tikzstyle{h} = [draw=black, dotted, line width=1.2]
 
    \node[node_stylea](x1) at (-11,0) {\tiny $1$};
    \node[node_stylek](x2) at (-10,0) {\tiny $2$};
    \node[node_stylek](x3) at (-9,0) {\tiny $18$};
    \node[node_stylek](x4) at (-8,0) {\tiny $17$};
    \node[node_stylek](x5) at (-7,0) {\tiny $7$};
    \node[node_stylek](x6) at (-6,0) {\tiny $8$};
    \node[node_stylea](x7) at (-5,0) {\tiny $3$};
    \node[node_stylek](x8) at (-4,0) {\tiny $4$};
    
    \node[node_styleh](h11) at (-3.5,3) {};
    \node[node_styleh](h12) at (-3.5,-3) {};
    
    \node[node_stylek](x9) at (-3,0) {\tiny $13$};  
    \node[node_stylek](x10) at (-2,0) {\tiny $14$};
    \node[node_stylea](x11) at (-1,0) {\tiny $5$};
    \node[node_stylek](x12) at (0,0) {\tiny $6$};
    \node[node_stylea](x13) at (1,0) {\tiny $9$};
    \node[node_stylek](x14) at (2,0) {\tiny $10$};
    \node[node_stylek](x15) at (3,0) {\tiny $15$};
    \node[node_stylek](x16) at (4,0) {\tiny $16$};
    
    \node[node_styleh](h21) at (4.5,3) {};
    \node[node_styleh](h22) at (4.5,-3) {};
    
    \node[node_stylek](x17) at (5,0) {\tiny $22$};
    \node[node_stylek](x18) at (6,0) {\tiny $21$};
    \node[node_stylek](x19) at (7,0) {\tiny $20$};
    \node[node_stylek](x20) at (8,0) {\tiny $19$};
    \node[node_stylea](x21) at (9,0) {\tiny $11$};
    \node[node_stylek](x22) at (10,0) {\tiny $12$};
    \node[node_stylek](x23) at (11,0) {\tiny $23$};
    \node[node_stylek](x24) at (12,0) {\tiny $24$};

    \node[node_styleh](l1) at (-7.5,-2) {$P_1$};
    \node[node_styleh](l2) at (0.5,-2) {$P_2$};
    \node[node_styleh](l3) at (8.5,-2) {$P_3$};
    
    \begin{pgfonlayer}{bg} 
    
    \draw[t,-]  (x1) to [out=285,in=255] (x2);
    \draw[t,-]  (x3) to [out=285,in=255] (x4);
    \draw[t,-]  (x5) to [out=285,in=255] (x6);
    \draw[t,-]  (x7) to [out=285,in=255] (x8);
    \draw[t,-]  (x9) to [out=285,in=255] (x10);
    \draw[t,-]  (x11) to [out=285,in=255] (x12);
    \draw[t,-]  (x13) to [out=285,in=255] (x14);
    \draw[t,-]  (x15) to [out=285,in=255] (x16);
    \draw[t,-]  (x17) to [out=285,in=255] (x18);
    \draw[t,-]  (x19) to [out=285,in=255] (x20);
    \draw[t,-]  (x21) to [out=285,in=255] (x22);
    \draw[t,-]  (x23) to [out=285,in=255] (x24);
    
    \node at (-10.5,-1) () {$x_1$};
    \node at (-8.5,-1) () {$x_2$};
    \node at (-6.5,-1) () {$x_3$};
    \node at (-4.5,-1) () {$x_4$};
    \node at (-2.5,-1) () {$x_5$};
    \node at (-0.5,-1) () {$x_6$};
    \node at (1.5,-1) () {$x_7$};
    \node at (3.5,-1) () {$x_8$};
    \node at (5.5,-1) () {$x_9$};
    \node at (7.5,-1) () {$x_{10}$};
    \node at (9.5,-1) () {$x_{11}$};
    \node at (11.5,-1) () {$x_{12}$};

    \draw[s,-]  (x1) to [out=75,in=105] (x6);
    \draw[s,-]  (x2) to [out=75,in=105] (x3);
    \draw[s,-]  (x4) to [out=75,in=105] (x5);

    \draw[s,-]  (x7) to [out=75,in=105] (x10);
    \draw[s,-]  (x8) to [out=75,in=105] (x9);
    \draw[s,-]  (x11) to [out=75,in=105] (x12);

    \draw[s,-]  (x7) to [out=75,in=105] (x10);
    \draw[s,-]  (x8) to [out=75,in=105] (x9);
    \draw[s,-]  (x11) to [out=75,in=105] (x12);

    \draw[s,-]  (x13) to [out=75,in=105] (x20);
    \draw[s,-]  (x14) to [out=75,in=105] (x15);
    \draw[s,-]  (x16) to [out=75,in=105] (x17);
    \draw[s,-]  (x18) to [out=75,in=105] (x19);
    
    \draw[s,-]  (x21) to [out=75,in=105] (x24);
    \draw[s,-]  (x22) to [out=75,in=105] (x23);

    \draw[h,-] (h11) to [out=270,in=90] (h12);
    \draw[h,-] (h21) to [out=270,in=90] (h22);
 
    \end{pgfonlayer}
    \end{tikzpicture}
\caption{Correctness of the hafnian computation: A canonical alternating cycle cover, partitioned in three balanced parts $P_1,P_2,$ and $P_3$, implicitly computed by the three subcircuits $C_1,C_2,$ and $C_3$ in Thm.~\ref{thm:faster-inclusion-exclusion}, respectively. Every cycle cover represents a unique perfect matching in the underlying complete input graph. The cycles in the cycle cover alternates between actual edges representing entries in the input matrix $A$ (thin edges above the vertices) and auxiliary pairing edges (red edges below the vertices). In the canonical ordering, the cycles are ordered after their anchor (gray), their lowest ranked vertex.}
	\label{fig:hafnian}
\end{figure}
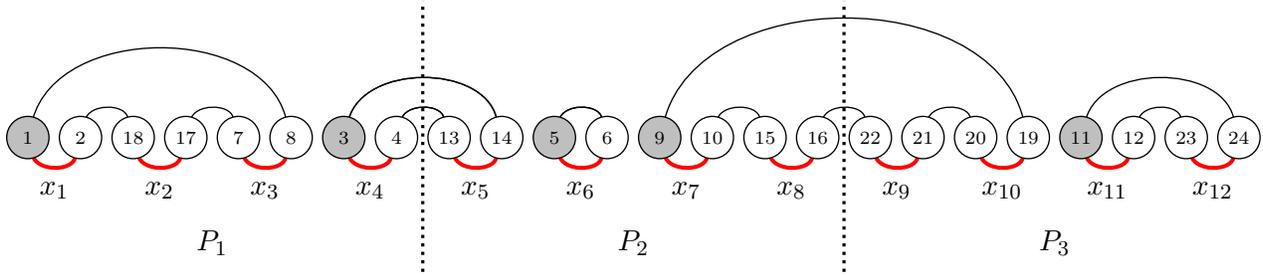

\medskip
\noindent
{\em Hafnian.}
As mentioned in the introduction, the hafnian of a symmetric matrix \(A \in \mathbb{F}^{2n \times 2n}\) is defined as 
$\haf A = \sum_{p \in P^2_{2n}} \prod_{(i, j) \in p} A_{i,j}$, where $P_{2n}^2$ is the set of all partitions of $[2n]$ into pairs.
This notion generalizes the permanent.  
In fact, for any square matrix $A$, we have
\begin{align*}
    \perm A = \haf \begin{pmatrix}
        0 & A \\ A^{\top} & 0
    \end{pmatrix}.
\end{align*}
In the following lemma, we present a generalization of Lemma~\ref{lem:permanent} to the hafnian.

\begin{Lem} \label{lem:hafnian}
Let $A \in \F^{2n \times 2n}$ be a symmetric matrix and $x = \{ x_1, \dots, x_n \}$. Then, the hafnian $\haf A$ can be computed as the coefficient of the monomial $\prod_{i=1}^n x_i$ in a polynomial $P(x)$ that can be computed by a polynomial-size $1$-skew arithmetic circuit.
\end{Lem}
\begin{proof}
    We loosely follow an algorithm by Cygan and Pilipczuk for computing the hafnian~\cite{CyganP2015}.  
    Construct a weighted multigraph $G$ on the vertex set $[2n]$ as follows. For each pair $i<j\in [2n]$, add a \emph{black} edge with weight given by the entry $A[i,j]$.
    In addition, for each $i\in [n]$, add a \emph{red} (pairing) edge connecting $2i-1$ and $2i$, and assign it the weight given by the indeterminate $x_i$.
    
    An \emph{alternating cycle cover} is a collection of cycles in $G$, each of which alternates between black and red edges, such that every vertex is incident to exactly one black edge and one red edge.
    As observed by Cygan and Pilipczuk~\cite{CyganP2015}, there is a bijection between alternating cycle covers and perfect matchings:
    given any perfect matching, one can add the red edges to obtain an alternating cycle cover, and vice versa.
    Let $\mathcal{A}$ denote the set of all alternating cycle covers of $G$. Then we have
    \begin{align} \label{eq:haf}
        \haf A \cdot \prod_{i \in [n]} x_i = \sum_{A \in \mathcal{A}} \prod_{e \in E(A)} A[e].
    \end{align}
    
    To facilitate the computation of the hafnian, we now introduce a canonical ordering for cycle covers.
    This ordering enables us to represent an alternating cycle cover as an ordered sequence of cycles rather than as an unordered set.
    For a given cycle cover, define the \emph{anchor} of each cycle to be its smallest vertex (with respect to the natural ordering).
    Then, order the cycles in increasing order of their anchors, and list the vertices within each cycle according to the order in which they are visited starting from the anchor and followed by a red edge (see Figure~\ref{fig:hafnian}).
    
    Furthermore, we define \emph{alternating clows}\footnote{The term ``clow'' (short for closed ordered walk) was coined by Mahajan and Vinay \cite{MahajanV97} in the context of combinatorial determinant computation.}, where the disjointness condition is relaxed.
    First, an \emph{alternating walk} is a sequence
    \[
    (i_0,\, e_1,\, i_1,\, e_2,\, \dots,\, e_s,\, i_s),
    \]
    where each $i_j$ is a vertex and each $e_j$ is an edge of $G$.
    The \emph{length} of an alternating walk is defined as $s$, and the walk is called \emph{closed} if $i_0=i_s$.
    The anchor of an alternating closed walk is defined as its smallest vertex.
    An \emph{alternating clow} is then a sequence of closed alternating walks $(W_1, \dots, W_k)$ such that in each $W_i$ the anchor appears exactly once, and the anchors of the walks occur in strictly increasing order.
    The length of an alternating clow is the sum of the lengths of the individual walks.
    Note that every alternating cycle cover is an alternating clow, and indeed, an alternating clow is an alternating cycle cover if and only if each red edge is traversed exactly once.
    
    Now, define the polynomial
    \[
    P(x) = \sum_{A \in \hat{\mathcal{A}}} \prod_{e \in E(A)} A[e],
    \]
    where $\hat{\mathcal{A}}$ is the collection of all alternating clow sequences.
    In particular, for $d_1, \dots, d_n \in \mathbb{N}$, the coefficient of $\prod_{i \in [n]} x_i^{d_i}$ in $P(x)$ is obtained by summing over all alternating clow sequences in which the red edge connecting $2i-1$ and $2i$ is traversed exactly $d_i$ times.
    Consequently, the coefficient of $\prod_{i \in [n]} x_i$ sums over all alternating cycle covers, and thus, by Equation~\eqref{eq:haf}, equals $\haf A$.
    
    We will show that $P(x)$ can be computed by a polynomial-size 1-skew arithmetic circuit.  
    The construction is inspired by the dynamic programming algorithm for determinant computation~\cite{MahajanV97,Rote01}.  
    For each edge $e$ in $G$, there is an input gate $u_e$ labeled by its weight.  
    We also introduce another input gate $v_{0,1,1}$ labeled by the constant 1.
    Moreover, for each $\ell, i, h \in [2n]$ with $i \ge h$, we create a sum gate $v_{\ell, i, h}$ that computes the sum over all partial clow sequences (i.e., sequences where the last alternating walk is not necessarily closed) of length $\ell$, such that the last vertex is $i$ and the anchor of the last alternating walk is $h$.  
    There are two ways to extend a partial clow sequence: either continue the current alternating walk or start a new one.  
    Accordingly, the gate $v_{\ell, i, h}$ is connected to nodes corresponding to these two cases (here, for simplicity, we allow it to have fan-in greater than two. It can be transformed into a circuit with fan-in two gates by introducing auxiliary nodes.):  
    \begin{itemize}
        \item For the former case, for each $j \in [2n]$ with $j > h$, we introduce an auxiliary product gate $w_{\ell,e,h}$, which is connected from $v_{\ell-1,i,h}$ and $u_e$, where $e = \{ i, j \}$ is a red edge if $\ell$ is odd and a black edge if $\ell$ is even.
        The gate $w_{\ell,e,h}$ then connects to $v_{\ell,i,h}$.
        \item For the latter case (applicable when $\ell$ is even), for each $i', h' \in [2n]$ with $h' < h$, we introduce a product gate $w_{\ell-1,i',h'}'$ connected from $v_{\ell-1,i',h'}$ and $u_e$, where $e = \{ i', h' \}$.
        The gate $w_{\ell-1,i',h'}'$ then connects to $v_{\ell,i,h}$.
    \end{itemize}
    It is straightforward to verify that the resulting arithmetic circuit is of polynomial size (more precisely, $O(n^4)$ fan-in two gates).  
    Moreover, the circuit is 1-skew.
\end{proof}

Lemma~\ref{lem:hafnian} provides a polynomial-size 1-skew arithmetic circuit in which the coefficient of \(\prod_i x_i\) equals the hafnian.  
By applying Theorem~\ref{thm:faster-inclusion-exclusion} to this circuit, we obtain Theorem~\ref{thm:hafnian} (see also Figure~\ref{fig:hafnian}).

\hafnian*

\medskip
\noindent
{\em Set partitioning.}

In the set partition problem, we are given a family of sets $\mathcal{F} \subseteq \binom{[n]}{q}$ and are tasked with finding a subfamily $\mathcal{F}' \subseteq \mathcal{F}$ that forms a partition of $[n]$.  
In the following lemma, we construct a $q$-skew arithmetic circuit that counts the number of such subfamilies.

\begin{Lem} \label{lem:set-partitioning}
Let $x = \{ x_1, \dots, x_n \}$.
For a set family $\mathcal{F} \subseteq \binom{[n]}{q}$, the number of subcollections $\mathcal{F}' \subseteq \mathcal{F}$ such that $\mathcal{F}'$ forms a partition of $[n]$ can be computed as the coefficient of the monomial $\prod_{i=1}^n x_i$ in a polynomial $P(x)$ that can be computed by a polynomial-size $q$-skew arithmetic circuit.
\end{Lem}

\begin{proof}
Consider the polynomial
\begin{align*}
P(x) = \prod_{S \in \mathcal{F}} \left( 1+ \prod_{i \in S} x_i \right).
\end{align*}
Since each term inside the product is a sum of monomials of degree at most $q$, the polynomial can be computed by a polynomial-size $q$-skew arithmetic circuit.

Expanding the product, we obtain
\begin{align*}
P(x) = \sum_{\mathcal{F}' \subseteq \mathcal{F}} \prod_{i = 1}^n x_i^{d_{i,\mathcal{F}'}},
\end{align*}
where $d_{i,\mathcal{F}'}$ denotes the number of sets $S \in \mathcal{F}'$ that contain element $i$.

To form a valid partition of $[n]$, each element $i \in [n]$ must appear in exactly one set in $\mathcal{F}'$, meaning $d_{i,\mathcal{F}'} = 1$ for all $i$. The coefficient of $\prod_{i=1}^n x_i$ thus counts the number of such valid partitions.
\end{proof}

Applying Theorem~\ref{thm:faster-inclusion-exclusion} (for the more general $q$-skew circuits; see the remark below the theorem) to the $q$-skew arithmetic circuit provided by Lemma~\ref{lem:set-partitioning} over a sufficiently large prime field (with $2^{\Theta(n^q)}$ elements) yields Theorem~\ref{thm:countsp}:

\setpartitioning*

\section{Applications to parameterized problems}

\label{sect:detection}

The power of Theorem~\ref{thm:faster-inclusion-exclusion} extends well beyond pure counting problems.  
In this section, we show that the same techniques can be used to speed up parameterized decision problems under the assumption that $\sigma(P_{\mathbb{N}}) < H(1/3)^{-1}$ over fields of characteristic 2.

For example, consider the task of \emph{multilinear detection}: given a polynomial $P(x_1,\ldots,x_n)$, decide whether its monomial expansion contains a monomial of degree $k$ that is multilinear, i.e., where every variable appears with degree at most one.  
This notion was introduced by Koutis~\cite{Koutis08} and subsequently refined by Williams and Koutis~\cite{KoutisW2016} using group algebra.  
Multilinear detection plays an important role in parameterized algorithms; indeed, it is well-known that the $k$-path problem --- where one seeks a path of length $k$ in a directed graph $G$ --- can be solved via multilinear detection (see Lemma~\ref{lem:k-path}).  
The fastest known randomized algorithm for the $k$-path problem relies on multilinear detection and runs in $O^*(2^k)$ time over fields of characteristic 2 \cite{Williams2009}.

Recently, Eiben et al.~\cite{EibenKW2024} introduced the \emph{determinantal sieving} method, which generalizes multilinear detection to linear matroids.  
In this section, we show that when the polynomial of interest is computed by a skew circuit, determinantal sieving can be performed more efficiently if $\sigma(P_{\mathbb{N}}) < H(1/3)^{-1}$ over fields of characteristic 2.

To present their result, for a monomial $m = x_1^{d_1} x_2^{d_2} \cdots x_n^{d_n}$, we define its support as $\supp(m) = \{ i \in [n] \mid d_i \ge 1 \}$.

\begin{Thm}[Determinantal sieving~\cite{EibenKW2024}] \label{thm:determinantal-sieving}
    Let $x = \{ x_1, \dots, x_n \}$, let $P(x)$ be a homogeneous polynomial (given via black-box access) of degree $k$ over a field $\F$ of characteristic~2 with at least $2k$ elements, and let $A \in \F^{k \times n}$ be a matrix.
    There is a randomized $O^*(2^k)$-time algorithm to test if there is a 
    term $m$ in the monomial expansion of $P(x)$ such that the 
    matrix $A[\cdot, \osupp(m)]$ is nonsingular.
\end{Thm}

Determinantal sieving generalizes multilinear detection when applied to \emph{Vandermonde matrices}:

\begin{Def}[Vandermonde matrix] \label{def:vandermonde}
    For $k \le n \in \mathbb{N}$, a $k \times n$ Vandermonde matrix (over a field with at least $n + 1$ elements) is defined by $A[i,j] = x_j^{i-1}$, where $x_1, \dots, x_n$ are distinct.  
\end{Def}

Since any $k\times k$ submatrix of a Vandermonde matrix is nonsingular, applying determinantal sieving in this case recovers the standard multilinear detection result over fields of characteristic 2.

We show that given a 1-skew arithmetic circuit, one can construct an arithmetic circuit performing the same task via Theorem~\ref{thm:faster-inclusion-exclusion}.
Thus, the running time of determinantal sieving can be also improved if $\sigma(P_{\mathbb{N}}) < H(1/3)^{-1}$ over a field of characteristic 2.

Eiben et al.~\cite{EibenKW2024} also notes a variant, dubbed \emph{odd sieving}, which tests for the presence of a term such that the associated submatrix of its odd support has full row rank. 
Here, for a monomial $m = x_1^{d_1} x_2^{d_2} \cdots x_n^{d_n}$, its odd support is defined as $\osupp(m) = \{ i \in [n] \mid d_i \equiv_2 1 \}$.

\begin{Thm}[Odd sieving \cite{EibenKW2024}]
    Let $x = \{ x_1, \dots, x_n \}$, let $P(x)$ be a polynomial (given via black-box access) of degree $k$ over a field $\F$ of characteristic~2 with at least $d+k$ elements, and let $A \in \F^{k \times n}$ be a matrix.
    There is a randomized $O^*(2^k)$-time algorithm to determine whether there is a 
    term $m$ in the monomial expansion of $P(x)$ such that the 
    matrix $A[\cdot, \osupp(m)]$ has full row rank.
\end{Thm}

This variant has been used to solve problems such as finding an $(s,t)$-path of length at least $k$ in an undirected graph in randomized $O^*(1.66^k)$ time.  
In this section, we show that odd sieving too can be implemented via Theorem~\ref{thm:faster-inclusion-exclusion} when the polynomial is computed by a 1-skew circuit.  
In particular, this leads to an algorithm for the undirected bipartite $k$-path problem running in randomized $O^*(2^{H(1/3)(\sigma(P_{\mathbb{N}}) + \varepsilon)k/2})$ time for all $\varepsilon > 0$.
Since the best known running time for this problem is $O^*(2^{k/2})$, a speedup is achieved if $\sigma(P_{\mathbb{N}}) < H(1/3)^{-1}$.

In this section, we assume that $\mathbb{F}$ is a field of characteristic 2 with sufficiently many elements (otherwise use an extension field).  
We work in the word-RAM model, where arithmetic operations over $\mathbb{F}$ can be performed in $O(1)$ time.

In Subsection~\ref{subsec:faster-determinantal-sieving}, we present how determinantal sieving can be implemented using Theorem~\ref{thm:faster-inclusion-exclusion}.
We provide several applications of this approach in Subsection~\ref{subsec:parameterized-applications}.

\subsection{Determinantal sieving via Theorem~\ref{thm:faster-inclusion-exclusion}} \label{subsec:faster-determinantal-sieving}

We begin with an implementation of Theorem~\ref{thm:determinantal-sieving} via Theorem~\ref{thm:faster-inclusion-exclusion}:

\begin{Thm} \label{thm:faster-determinantal-sieving}
    Let $x$ be a set of variables, let $P(x)$ be a polynomial of degree $k$ over a field $\F$ of characteristic~2 with at least $2k$ elements, and let $A \in \F^{k \times n}$ be a matrix.
    Suppose that a polynomial-size 1-skew arithmetic circuit that computes $P(x)$ is given.
    Then, for every $\varepsilon > 0$, there is a randomized $O(2^{H(1/3)(\sigma(P_\mathbb{N})+\varepsilon)k})$-time algorithm to determine whether there is a term $m$ in the monomial expansion of $P(x)$ such that the 
    matrix $A[\cdot, \supp(m)]$ is nonsingular.
\end{Thm}

To prove Theorem~\ref{thm:faster-determinantal-sieving}, we revisit the proof of the original determinantal sieving result by Eiben et al.~\cite{EibenKW2024}.  
Upon close examination, the main idea can be summarized as follows:

\begin{Lem}[\cite{EibenKW2024}] \label{lem:transformation}
Let $x = \{ x_1, \dots, x_n \}$ and $y = \{ y_1, \dots, y_k \}$ be sets of variables, and let $\mathbb{F}$ be a field of characteristic 2. Let $P(x)$ be a homogeneous polynomial of degree $k$ over $\mathbb{F}$, and let $A \in \mathbb{F}^{k \times n}$ be a matrix. 

Consider the substitution
\[
x_i \ \rightarrow \ x_i \left(\sum_{j=1}^k y_j\, A[j,i]\right) \quad \text{for each } i \in [n].
\]
Let $Q(x,y)$ denote the resulting polynomial in the variables $x \cup y$. Then the coefficient of the monomial 
$\prod_{j=1}^k y_j$
in $Q(x,y)$ is given by
\[
P^{\star}(x) = \sum_{m} c_m \cdot \det\bigl(A[\cdot, \supp(m)]\bigr) \cdot m,
\]
where the sum ranges over all multilinear monomials $m$ in $P(x)$ and $c_m$ is the coefficient of $m$ in $P(x)$.
\end{Lem}

We remark that the determinantal sieving result of Theorem~\ref{thm:determinantal-sieving} follows from Lemma~\ref{lem:transformation}, the inclusion-exclusion principle, and the Schwartz-Zippel lemma.  
Recall that the Schwartz-Zippel lemma states that a nonzero polynomial of degree $d$ over a field $\mathbb{F}$ evaluates to a nonzero value with probability at least $1 - d/|\mathbb{F}|$ when the coordinates are chosen uniformly at random.  
In particular, by the Schwartz-Zippel lemma, it suffices to evaluate $P^{\star}(x)$ at random coordinates.  
Moreover, by the inclusion-exclusion principle, $P^{\star}(x)$ can be expressed as a sum of $2^k$ evaluations of $Q$.  
Thus, Theorem~\ref{thm:determinantal-sieving} follows.

We are now ready to prove Theorem~\ref{thm:faster-determinantal-sieving}.  

\begin{proof}[Proof of Theorem~\ref{thm:faster-determinantal-sieving}]
    For each variable $x_i$, we choose a random element from $\mathbb{F}$ and apply the substitution as described in Lemma~\ref{lem:transformation}.  
    We are interested in the coefficient of $\prod_{j=1}^k y_j$.  
    Observe that $P^{\star}$ is a nonzero polynomial if and only if $P$ contains a multilinear monomial $m$ such that $A[\cdot, \supp(m)]$ is nonsingular.  
    By the Schwartz-Zippel lemma (and using that $|\mathbb{F}| \ge 2k$), if $P$ contains a desired monomial then the probability that $P^{\star}(X)$ evaluates to zero is at most $\frac{1}{2}$.
    
    Since all the substituted polynomials in the computation of $P^{\star}$ are of degree 1, it follows that $P^{\star}$ can be computed by a polynomial-size 1-skew arithmetic circuit.  
    Consequently, by Theorem~\ref{thm:faster-inclusion-exclusion}, we can evaluate $P^{\star}(X)$ in time 
    $O^*(2^{H(1/3)(\sigma(P_{\mathbb{N}})+\varepsilon)k})$.
    Thus, with high probability, we can test whether $P$ contains a monomial whose support forms a basis of $A$ within the stated time bound.
\end{proof}

We now turn to the odd sieving method.  
A careful examination of the proof by Eiben et al.~\cite{EibenKW2024} reveals the following:

\begin{Lem}[\cite{EibenKW2024}] \label{lem:transformation-odd}
    Let $x = \{ x_1, \dots, x_n \}$, $x' = \{ x_1', \dots, x_n' \}$, $y = \{ y_1, \dots, y_k \}$ be sets of variables, let $z$ be a variable, and let $\mathbb{F}$ be a field of characteristic 2. Let $P(x)$ be a polynomial of degree $d$ over $\mathbb{F}$, and let $A \in \mathbb{F}^{k \times n}$ be a matrix. 
    Consider the substitution
    \[
    x_i \ \rightarrow \ x_i \left(1 + z x_i' \sum_{j=1}^k y_j\, A[j,i]\right) \quad \text{for each } i \in [n].
    \]
    Let $Q(x,y,z)$ denote the resulting polynomial in the variables $x \cup y \cup \{ z \}$. Then the coefficient of the monomial 
    $z^{k} \prod_{j=1}^k y_j$
    in $Q(x,y,z)$ is given by
    \[
    P^{\star}(x) = \sum_{m} c_m \cdot \left( \sum_{\mu} \det A[\cdot, \supp(\mu)] \cdot \mu' \right) \cdot m,
    \]
    where the outer sum ranges over all monomials $m$ in $P(x)$ with coefficient $c_m$, the inner sum ranges over all multilinear monomials $\mu$ of degree $k$ such that $\supp(m) \subseteq \osupp(m)$, and $\mu'$ denotes the monomial $\prod_{i \in \supp(\mu)} x_i'$.
\end{Lem}

With Lemma~\ref{lem:transformation-odd} in hand, one can show (in a manner analogous to Theorem~\ref{thm:faster-determinantal-sieving}) that the odd sieving variant can be implemented via a 1-skew circuit
with the only difference that Lemma~\ref{lem:transformation-odd} refers to the coefficient of $z^k$.
This can be achieved by homogenization (Lemma~\ref{lem:homogenization}). 

\begin{Thm} \label{thm:faster-odd-sieving}
    Let $x$ be a set of variables, let $P(x)$ be a homogeneous polynomial of degree $k$ over a field $\F$ of characteristic~2 with at least $d+k$ elements, and let $A \in \F^{k \times n}$ be a matrix.
    Suppose that a polynomial-size 1-skew arithmetic circuit that computes $P(x)$ is given.
    Then, for every $\varepsilon > 0$, there is a randomized $O(2^{H(1/3)(\sigma(P_\mathbb{N})+\varepsilon)k})$-time algorithm to determine whether there is a term $m$ in the monomial expansion of $P(x)$ such that the 
    matrix $A[\cdot, \osupp(m)]$ has full row rank.
\end{Thm}

\subsection{Applications} \label{subsec:parameterized-applications}

We present three applications of Theorem~\ref{thm:faster-determinantal-sieving}: the (directed) $k$-path problem, the 3-dimensional matching problem (and, more generally, the 3-matroid intersection problem), and the long cycle problem on bipartite graphs.  
The first two applications rely on Theorem~\ref{thm:determinantal-sieving}, whereas the last one utilizes Theorem~\ref{thm:faster-odd-sieving}.  
It is worth noting that these examples are not exhaustive; a broader class of combinatorial problems can benefit from our approach.  
Indeed, many of the problems mentioned in Eiben et al.~\cite{EibenKW2024} can be addressed via either Theorem~\ref{thm:determinantal-sieving} or Theorem~\ref{thm:faster-odd-sieving}.

To apply either of these theorems, it suffices to provide a polynomial-size 1-skew arithmetic circuit for the given problem, along with an appropriate matrix encoding the problem's structure.  
Although the circuit constructions for our three applications are already known implicitly~\cite{BjorklundHKK17,EibenKW2024,Williams2009}, we present them here for completeness.

\medskip
\noindent
{\em $k$-path.}
Recall that in the $k$-path problem, the task is to find a path of length $k$ in a given directed graph.  
A randomized $O^*(2^k)$-time algorithm is known for this problem \cite{Williams2009}.
The following lemma provides a 1-skew circuit that can be used to solve the $k$-path problem.

\begin{Lem} \label{lem:k-path}
    Let $G = (V, E)$ be a directed graph and let $x = \{ x_v \mid v \in V \}$ be a set of variables.
    Then there exists a homogeneous polynomial $P(x)$ of degree $k+1$ over a field $\mathbb{F}$ of characteristic~2 (with at least $\Omega(k)$ elements) such that $P(x)$ contains a multilinear term if and only if $G$ contains a path of length $k$.  
    Moreover, a 1-skew circuit computing $P(x)$ can be constructed in randomized polynomial time.
\end{Lem}

\begin{proof}
    For each $i \in [k]$ and for each edge $e \in E$, introduce an indeterminate $y_{i,e}$.  
    For a walk $W = (u_0, e_1, u_1, e_2, \dots, e_{k}, u_k)$ of length $k$, define its \emph{labeled walk monomial} as
    \[
    m(W) = \left( \prod_{i = 0}^k x_{u_i} \right) \left( \prod_{i = 1}^k y_{i,e_i} \right).
    \]
    Next, define the \emph{labeled walk polynomial} of $G$ by
    \[
    Q(x, y) = \sum_{W} m(W),
    \]
    where the sum is taken over all walks $W$ of length $k$ in $G$.  
    Since every walk yields a distinct monomial in $Q(x,y)$, no cancellation occurs.
    
    Observe that a monomial in $Q(x,y)$ is multilinear in the variables $X$ if and only if the corresponding walk is simple (i.e., a path).  
    Thus, $Q(x,y)$ contains a multilinear term in $x$ if and only if $G$ contains a path of length $k$.
    
    We now obtain the desired polynomial $P(x)$ by substituting random field elements for each $y_{i,e}$:  
    By the Schwartz-Zippel lemma (using that $|\mathbb{F}| \ge 2k$), if $Q(x, y)$ contains a multilinear term then, after the substitution, $P(X)$ retains a nonzero multilinear term with probability at least $1/2$.    

    It remains to describe an arithmetic circuit that computes $P(x)$.  
    For each $i \in [k]$, define a $|V| \times |V|$ matrix $A_i \in \mathbb{F}[X]^{V \times V}$ by
    \[
    A_i[u,w] = \begin{cases}
    \hat{y}_{i,(u,w)} \, x_u & \text{if } (u,w) \in E,\\[1mm]
    0 & \text{otherwise},
    \end{cases}
    \]
    where $\hat{y}_{i,(u,w)}$ denotes the randomly chosen value for $y_{i,(u,w)}$.  
    Also, define a vector $\alpha \in \mathbb{F}[X]^{V}$ by setting 
    $\alpha[w] = x_w$ for each $w \in V$,
    and let $1_V \in \mathbb{F}[X]^{V}$ be the all-ones vector.  
    Then we have
    \[
    P(X) = 1_V^{\top} \, A_1 A_2 \cdots A_k \, \alpha.
    \]
    Since each matrix $A_i$ has entries of degree at most one, it follows that $P(X)$ is computed by a polynomial-size 1-skew arithmetic circuit, which can be constructed in randomized polynomial time.
\end{proof}

By applying Theorem~\ref{thm:faster-determinantal-sieving} to the circuit from Lemma~\ref{lem:k-path}, and using a $(k+1) \times n$ Vandermonde matrix, we obtain the following:

\kpath*

\medskip
\noindent
{\em 3-dimensional matching and 3-linear matroid intersection.}
The next target problem is the \emph{3-dimensional matching} problem, defined as follows.  
Given three sets $U$, $V$, and $W$, a collection of triplets $\mathcal{E} \subseteq U \times V \times W$, and an integer $k \in \mathbb{N}$, the goal is to find $k$ pairwise disjoint triplets $\mathcal{M} \subseteq \mathcal{E}$.  
It is known that this problem can be solved in randomized $O^*(2^k)$ time~\cite{Bjorklund2010}.

In fact, we consider a more general problem called \emph{3-linear matroid intersection}.  
In this problem, we are given three matrices $A$, $B$, and $C \in \mathbb{F}^{k \times m}$ over a field $\mathbb{F}$ (with characteristic~2), and the task is to find a set $S \subseteq \binom{[m]}{k}$ such that the submatrices $A[\cdot,S]$, $B[\cdot,S]$, and $C[\cdot,S]$ are all nonsingular.  
Using the determinantal sieving technique~\cite{EibenKW2024}, it has been shown that this problem can be solved in randomized $O^*(2^k)$ time.

The 3-linear matroid intersection problem generalizes the 3-dimensional matching problem via Vandermonde matrices (see Definition~\ref{def:vandermonde}) as follows.  
Suppose that $m = |\mathcal{E}|$.  
Let $M$ be a $k \times n$ Vandermonde matrix.
We construct three $k \times m$ matrices $A$, $B$, and $C$, where for each $j \in [m]$, the $j$\textsuperscript{th} column of $A$ (respectively, $B$, $C$) is taken to be the $x_j$\textsuperscript{th} (respectively, $y_j$\textsuperscript{th}, $z_j$\textsuperscript{th}) column of $M$, where $(x_j, y_j, z_j)$ is the $j$\textsuperscript{th} triplet of $\mathcal{E}$.  
The equivalence between these instances is straightforward to verify.

The next lemma shows that a generating polynomial for 2-linear matroid intersection can be computed by a 1-skew circuit, which leads to a faster algorithm for 3-linear matroid intersection (assuming $\sigma(P_{\mathbb{N}}) < H(1/3)^{-1}$).

\begin{Lem} \label{lem:3-matroid-intersection}
    Let $x = \{ x_1, \dots, x_m \}$ be a set of variables, and let $A, B \in \mathbb{F}^{k \times m}$ be matrices over a field $\mathbb{F}$.  
    Then there exists a homogeneous polynomial $P(x)$ of degree $k$ such that for every $S \in \binom{[m]}{k}$, $P(X)$ contains a multilinear term $\prod_{i \in S} x_i$ if and only if both $A[\cdot, S]$ and $B[\cdot, S]$ are nonsingular.  
    Moreover, a 1-skew circuit computing $P(x)$ can be constructed in polynomial time.
\end{Lem}
    
    \begin{proof}
    Define
    \[
    P(x) = \det \Bigl( A \cdot \diag(x_1, \dots, x_m) \cdot B^\top \Bigr).
    \]
    By the Cauchy-Binet formula, we have
    \[
    P(x) = \sum_{S \in \binom{[m]}{k}} \det A[\cdot, S] \cdot \det B[\cdot, S] \cdot \prod_{i \in S} x_i,
    \]
    which immediately implies that $P(x)$ contains the multilinear term $\prod_{i \in S} x_i$ if and only if both $\det A[\cdot, S]$ and $\det B[\cdot, S]$ are nonzero, i.e., if and only if both $A[\cdot, S]$ and $B[\cdot, S]$ are nonsingular.
    
    Note that the matrix $A \cdot \diag(x_1, \dots, x_m) \cdot B^\top$ has entries of degree at most one in the variables~$x_i$.  
    Since the determinant of a matrix can be computed by a polynomial-size skew circuit~\cite{MahajanV97}, it follows that $P(x)$ can be computed by a polynomial-size 1-skew circuit.
\end{proof}

To solve the 3-matroid intersection (and 3-dimensional matching) problem, we apply Theorem~\ref{thm:faster-determinantal-sieving} to the circuit provided by Lemma~\ref{lem:3-matroid-intersection}, using the third matrix $C$.  
This yields the following:

\begin{Thm}
    For every $\varepsilon > 0$, there exists a randomized $2^{H(1/3)(\sigma(P_\mathbb{N})+\epsilon)k}$-time algorithm that,
    \begin{itemize}
        \item 
        given a collection $\mathcal{E} \subseteq U \times V \times W$ of triplets, decides whether $\mathcal{E}$ contains $k$ pairwise disjoint triplets, and
        \item 
        given three matrices $A, B, C \in \mathbb{F}^{k \times m}$, decides whether there exists a set $S \subseteq [m]$ with of size $k$
        such that the submatrices $A[\cdot,S]$, $B[\cdot,S]$, and $C[\cdot,S]$ are all nonsingular.
    \end{itemize}
\end{Thm}

\noindent
{\em Long cycle.}
The long cycle problem is defined as follows.  
Given a graph $G$ and an integer $k \in \mathbb{N}$, the task is to find a cycle of length \emph{at least} $k$.  
Note that this problem generalizes the Hamiltonicity problem when $k = n$.
It is known that this problem can be solved in randomized $O^*(2^{k/2})$ time on bipartite graphs, and in randomized $O^*(1.657^k)$ time on general undirected graphs~\cite{EibenKW2024}.  
In this paper, we assume that the input graph is an undirected bipartite graph.
We remark, however, that their random bipartitioning argument for general undirected graphs most likely extends to our setting as well.

We use a polynomial construction given by Eiben et al.~\cite{EibenKW2024}.

\begin{Lem}[\cite{EibenKW2024}] \label{lem:long-cycle}
    Let $G = (V, E)$ be an undirected graph, let $s, t \in V$ be two nonadjacent vertices, and let $x = \{ x_e \mid e \in E \}$ be a set of variables.
    Then there exists a polynomial $P(x)$ over a field $\mathbb{F}$ of characteristic 2 (with at least $\Omega(n)$) elements such that
    there is an $(s,t)$-path $\pi$ if and only if $P(x)$ contains a term $m$ with $E_{\pi} \subseteq \osupp(m)$, where $E_{\pi}$ is the edge set of $\pi$.
\end{Lem}

We now briefly describe the construction of the polynomial; for the full proof, see Eiben et al.~\cite{EibenKW2024}.

\begin{proof}[Proof sketch.]
    We define a matrix $A$ whose rows and columns are indexed by $V$, where 
    \begin{align*}
        A[u,w] = \begin{cases}
            x_{e} & \text{ if $e = \{ u, w \} \in E$} \\
            0 & \text{ otherwise},
        \end{cases}
    \end{align*}
    with the exception that $A[t, s] = 1$.
    Then, $P(x) = \det A$ is a desired polynomial.
    Since every entry of $A$ is of degree at most 1, $P(x)$ can be computed by a 1-skew circuit~\cite{MahajanV97}.
\end{proof}

We apply Theorem~\ref{thm:faster-odd-sieving} to the circuit from Lemma~\ref{lem:long-cycle} to prove the following:

\begin{Thm}
    For all $\varepsilon >0$, there is a randomized algorithm that, given an undirected bipartite graph $=(V,E)$, decides whether $G$ contains a cycle of length at least $k$ in $O^*(2^{H(1/3)(\sigma(P_\mathbb{N})+\epsilon)k/2})$ time.
\end{Thm}
\begin{proof}
    To solve the long cycle problem, it suffices to determine whether there exists an $(s,t)$-path of length at least $k$ for each edge $\{s,t\} \in E$.
    
    Let $(U, W)$ be a bipartition of $V$.
    We will construct $\frac{k}{2} \times m$ matrix $A$, where $m = |E|$.
    To that end, let $M$ a $\frac{k}{2} \times |U|$ Vandermonde matrix.
    For the $i$\textsuperscript{th} edge $e = \{ u, w \}$ with $u \in U$ and $w \in W$, the $i$\textsuperscript{th} column of $A$ is defined as the column of $U$ corresponding to $u$.
    
    We claim that in the polynomial from Lemma~\ref{lem:long-cycle}, $G$ has an $(s,t)$-path of length at least $k$ if and only if there is a term in $P(X)$ such that the submatrix of $A$ restricted to its odd support has full row rank.
    One direction is clear---if there is an $(s,t)$-path of length at least $k$, say $(s = u_1, v_1, u_2, v_2, \dots, u_{\ell}, v_{\ell} = t)$, then the corresponding submatrix indexed by $\{ \{ u_i, v_i \} \mid i \in [\frac{k}{2}] \}$ has full row rank.
    Conversely, if there is a term whose odd support yields a full row rank submatrix, then the corresponding edges cover at least $k/2$ vertices of $U$.
    This implies that there is an $(s,t)$-path of length at least $k$.

    Consequently, by applying Theorem~\ref{thm:faster-odd-sieving}, we obtained an algorithm with the stated time bound.
\end{proof}

\section{Applications to Hamiltonicity Parameterized by Treewidth}

\label{sect:treewidth}

In this section we relate the complexity of the Hamiltonicity problem on graphs with given tree decomposition of small width to the tensor rank of sequence of three-tensor called the \emph{matchings connectivity tensors}, defined as follows. For convenience, we let $U:=\{1,\ldots,q\}$.
\begin{Def}[Fingerprint]
	A \emph{$U-$fingerprint} is a pair $(d,M)$ where $d: U \rightarrow \{0,1,2\}$ and $M$ is a perfect matching of $Z_{d^{-1}(1)}$. 
\end{Def}

\begin{Def}[Matchings Connectivity Tensor]\label{def:matchtensor}
	For an integer $q:=|U|$, we define the \emph{matchings connectivity tensor} $H_q$ as
	\[
	H_q(x,y,z) =	\sum_{\substack{(d_1,M_1),(d_2,M_2),(d_3,M_3) \\ \forall v \in [q]: d_1(v)+d_2(v)=d_3(v) \\ M_1 \cup M_2 \cup M_3 \text{ is a cycle}} } x_{d_1,M_1}\cdot y_{d_2,M_2}\cdot z_{d_3,M_3},
	\]
	where the sum runs over all $U$-fingerprints $(d_1,M_1),(d_2,M_2),(d_3,M_3)$.
\end{Def}

Note that $H_{q'}(x,y,z)$ is a sub-tensor of $H_{q}(x,y,z)$ whenever $q' < q$ since we can restrict $H_{q}(x,y,z)$ to $U$-fingerprints satisfying $d_1(e)=2$ and $d_2(e)=d_3(e)=0$ for $e \in \{q'+1,\ldots,q\}$.

\begin{Thm}\label{thm:hcmain}
	For all $\varepsilon >0$, there is a randomized algorithm that takes an $n$-vertex graph $G$ along with a tree decomposition $\mathbb{T}$ of $G$ of treewidth $\tw$ as input, and outputs whether $G$ has a Hamiltonian cycle in time $O^*\left((2+\sqrt{2})^{(\sigma(H_\mathbb{N})+\varepsilon)\tw}\right)$.
\end{Thm}

The proof of Theorem~\ref{thm:hcmain} continues in a natural way an approach used by Cygan et al.~\cite{CyganKN2018} that gave an $O^*((2+\sqrt{2})^{\pw})$ time algorithm for the Hamiltonicity problem when given a path decomposition of pathwidth $\pw$ by relating it to the rank of a matrix that indicates whether the union of two perfect matchings is a cycle, the so-called matchings connectivity \emph{matrix}.

Before we present our approach, we present preliminaries on tree decompositions in Subsection~\ref{subsec:tw} and preliminaries on the matchings connectivity matrix in~\ref{subsec:mcm}. Afterwards in Subsection~\ref{subsec:revckn}, we revisit the dynamic programming approach of~\cite{CyganKN2018} for the Hamiltonicity problem parameterized by pathwidth and discuss what needs to be done to extend it to a dynamic programming algorithm parameterized by treewidth that establishes~Theorem~\ref{thm:hcmain}.
Then we show in Subsection~\ref{subsec:mct} that the matchings connectivity matrix can be decomposed into Kronecker products in a way central to this paper. Finally, we prove Theorem~\ref{thm:hcmain} in Subsection~\ref{subsec:fin} by combining the previous parts.

\subsection{Standard definitions related to treewidth}
\label{subsec:tw}
Throughout this section we fix the input graph $G$ and its tree decomposition $\mathbb{T}$, and we assume $\tw$ is the treewidth of $\mathbb{T}$.

\begin{Def}[Tree Decomposition, \cite{rs:minors3}]
	A \emph{tree decomposition} of an undirected graph~$G=(V,E)$ is a tree~$\treedecomp$ in which each 
	node~$i \in \treedecomp$ has an assigned set of vertices~$B_i \subseteq V$ 
	(called a \emph{bag}) such that $\bigcup_{x \in \treedecomp} B_i = V$ with the 
	following properties:
	\begin{itemize}
		\item for any $uv \in E$, there exists an~$i \in \treedecomp$ such that 
		$u,v \in B_i$, and
		\item if $v \in B_i$ and $v \in B_j$, then $v \in B_{j'}$ for all $j$ on 
		the path from $x$ to $y$ in $\treedecomp$.
	\end{itemize}
\end{Def}

Similarly, a \emph{path decomposition} is a tree decomposition with the additional property that $\treedecomp$ is a path.
In what follows we identify nodes of $\treedecomp$ and the bags assigned to them.
The \emph{width} of a tree decomposition~$\treedecomp$ is the size of the largest bag of $\treedecomp$ minus one, and the treewidth of a graph $G$ is the 
minimum width over all possible tree decompositions of~$G$.

We use the following definition of a nice tree decomposition:

\begin{Def}[Nice Tree Decomposition] \label{def:nicetreedecomp}
	A \emph{nice tree decomposition} is a tree decomposition with one special 
	bag $\rootv$ called the \emph{root} with $B_\rootv = \emptyset$ and in which each 
	bag is one of the following types:
	\begin{itemize}
		\item \textbf{Leaf bag}: a leaf $i$ of $\treedecomp$ with $B_i = \emptyset$.
		\item \textbf{Introduce vertex bag}: an internal vertex~$i$ of $\treedecomp$ 
		with one child vertex~$j$ for which $B_i = B_j \cup \{v\}$ 
		for some $v \notin B_j$. 
		This bag is said to \emph{introduce} $v$.
		\item \textbf{Introduce edge bag}: an internal vertex~$i$ of $\treedecomp$ labeled 
		with an edge $uv \in E$ with one child bag~$j$ for which 
		$u,v \in B_i = B_j$. 
		This bag is said to \emph{introduce} $uv$.
		\item \textbf{Forget bag}: an internal vertex~$i$ of $\treedecomp$ with one child 
		bag~$j$ for which $B_i = B_j \setminus \{v\}$ for some $v \in B_j$. 
		This bag is said to \emph{forget} $v$.
		\item \textbf{Join bag}: an internal vertex $i$ with two child vertices 
		$j$ and $j'$ with $B_i = B_{j'} = B_{j'}$.
	\end{itemize}
	We additionally require that every edge in $E$ is introduced exactly once. 
\end{Def}

This definition can be found in i.e.~\cite{CyganNPPRW22}.
Given a tree decomposition, a nice tree decomposition 
of equal width can be found in polynomial time (see~\cite{CyganNPPRW22}). Similarly we can convert any path decomposition into a nice tree decomposition of equal width in polynomial time, where a nice path decomposition means there are only leaf, introduce vertex, introduce edge, and forget bags.

For two bags $i,j$ of a rooted tree we say that $j$ is a descendant
of $i$ if it is possible to reach $i$ when starting at $j$ and going
only up (i.e. towards $\rootv$) in the tree.
In particular $i$ is its own descendant.
By fixing the root of $\treedecomp$, we associate with each bag $i$ 
in a tree decomposition $\treedecomp$ a vertex set $V_i \subseteq V$ where a vertex 
$v$ belongs to $V_i$ if and only if there is a bag $j$
which is a descendant of $i$ in $\treedecomp$ with $v \in B_j$.
We also associate with each bag $i$ of $\treedecomp$ a subgraph of $G$ as follows:

\[ G_i = \Big{(}V_i, E_i = \{e :  \textrm{$e$ is introduced in a descendant of $i$ } \}\Big{)}. \]

\subsection{Preliminaries on the Matchings Connectivity Matrix}\label{subsec:mcm}
If $X$ is a set, we let $K_X$ denote the complete graph with vertex set $X$.
If $a$ is a binary string, we let $\overline{a}$ denote the complement of $a$ (i.e. $\overline{a}_i=1-a_i$ for every $i$).

\begin{Def}[Basis matchings, paraphrased from Section 3.1 in \cite{CyganKN2018}]\label{def:basis}
For a set $X=\{x_1,\ldots,x_q\} \subseteq \mathbb{N}$ with even $q$, we define the graph $Z_X$ as follows: The vertex set $V(Z_X)$ is defined as $X$ and the edge set is defined as
\[
	E(Z_X) = \{ \{x_i,x_j \} : \lfloor j / 2\rfloor = \lfloor i / 2 \rfloor +1 \}.
\]
The graph $Z_X$ has $2^{q/2-1}$ perfect matchings, and we index them with $a \in \{0,1\}^{q/2-1}$ as follows:

\[
\begin{aligned}
	\mathcal{B}(X,a0) &:= \mathcal{B}(\{x_1,\ldots,x_{q-2}\},a) &\cup&\ \{ \{x_{q-1},x_q\}\}\\
	\mathcal{B}(X,a1) &:= \mathcal{B}(\{x_1,\ldots,x_{q-3},x_{q-1}\},a) &\cup&\ \{ \{x_{q-2},x_q\}\}.
\end{aligned}
\]
\end{Def}
See Figure~\ref{fig:z} for an example of a graph $Z_X$.

The \emph{matchings connectivity matrix} is a binary matrix indexed by all perfect matchings of $K_X$ that indicates whether two perfect matchings form a Hamiltonian cycle of $K_X$. The reason why the family of perfect matchings of $Z_X$ is referred to as \emph{basis matchings} is because of the following lemma that shows that, in the field $\mathbb{F}_2$, they form a basis of the matchings connectivity matrix:

\begin{Lem}[Theorem 3.4 in~\cite{CyganKN2018}]\label{lem:fact}
	If $M_1,M_2$ are perfect matchings of $K_{X}$, then
	\[
	 [M_1 \cup M_2 \text{ is a HC } ] \equiv_2 \sum_{a \in \{0,1\}^{|X|/2-1}} [M_1 \cup \mathcal{B}(X,a) \text{ is a HC }] \cdot [M_2 \cup \mathcal{B}(X,\overline{a}) \text{ is a HC }].
	\]
\end{Lem}
Here we use $\equiv_2$ to indicate the parities of the two quantities are equal and use Iverson's bracket notation $[b]$ to indicate $1$ if the Boolean $b$ is true and to indicate $0$ otherwise.
\begin{figure}
\begin{tikzpicture}[scale=0.5,shorten >=1pt, auto, node distance=1cm, ultra thick]
 
    \tikzstyle{node_stylek} = [circle, draw=black, fill=white, inner sep=1pt, minimum size=28pt, line width=.5]
  
    \tikzstyle{t} = [draw=black, line width=0.5]
 
    \node[node_stylek](x1) at (-6,0) {\huge $x_1$};
    \node[node_stylek](x2) at (-4,2.2) {\huge $x_2$};
    \node[node_stylek](x3) at (-4,-2.2) {\huge $x_3$};
    \node[node_stylek](x4) at (0,2.2) {\huge $x_4$};
    \node[node_stylek](x5) at (0,-2.2) {\huge $x_5$};
    \node[node_stylek](x6) at (4,2.2) {\huge $x_6$};
    \node[node_stylek](x7) at (4,-2.2) {\huge $x_7$};
    \node[node_stylek](x8) at (6,0) {\huge $x_8$};
    
    \begin{pgfonlayer}{bg} 
    
    \draw[t,-]  (x1.center) to [out=45,in=-135] (x2.center);
    \draw[t,-]  (x1.center) to [out=-45,in=135] (x3.center);
    \draw[t,-]  (x2.center) to [out=0,in=180] (x4.center);
    \draw[t,-]  (x2.center) to [out=-45,in=135] (x5.center);
    \draw[t,-]  (x3.center) to [out=45,in=-135] (x4.center);
    \draw[t,-]  (x3.center) to [out=0,in=180] (x5.center);
    \draw[t,-]  (x4.center) to [out=0,in=180] (x6.center);
    \draw[t,-]  (x4.center) to [out=-45,in=135] (x7.center);
    \draw[t,-]  (x5.center) to [out=45,in=-135] (x6.center);
    \draw[t,-]  (x5.center) to [out=0,in=180] (x7.center);
    \draw[t,-]  (x6.center) to [out=-45,in=135] (x8.center);
    \draw[t,-]  (x7.center) to [out=45,in=-135] (x8.center);
    
    \end{pgfonlayer}
    \end{tikzpicture}
\caption{The graph $Z_{X}$ where $X=\{x_1,\ldots,x_8\}$ with $x_1 < x_2 < \ldots < x_8$.}
	\label{fig:z}
\end{figure}
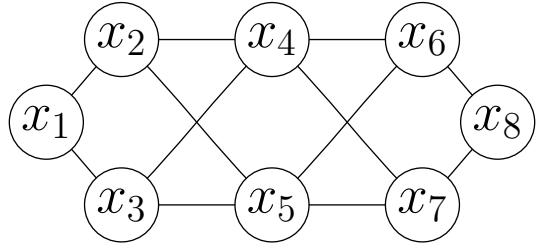

\subsection{The algorithm for Hamiltonicity parameterized by pathwidth~\cite{CyganKN2018}.}\label{subsec:revckn}
We paraphrase the algorithm from~\cite{CyganKN2018}. That algorithm uses a standard technique that assigns a random weight $\omega(e) \in \{1,\ldots,\omega_{\max}\}$ with $\omega_{\max}=n^2$ to every edge of the input graph and computes the parity of the number of Hamiltonian cycles $C$ with weight $\omega(C):=\sum_{e \in C}\omega(e)=w$ for every $w \in \{0,\ldots,n \cdot \omega_{\max} \}$. By the Isolation Lemma~\cite{MulmuleyVV87}, one of these parities is odd with constant probability if a Hamiltonian cycle exists (and otherwise all computed parities naturally are even). Hence, computing these parities is sufficient for obtaining a randomized algorithm for the decision variant of the Hamiltonicity problem.

For each bag $i$, we compute table entries~$t_i[d,w,M] \in \mathbb{Z}_2$ for all functions $d: B_i  \rightarrow \{0,1,2\}$, all
integers~$w\in\{0,\ldots,n\cdot \omega_{\max}\}$, and all perfect matchings~$M$ of $Z_{d^{-1}(1)}$. 
First, define $\mathcal{T}_i[d,w]$ as the family of edge sets $X \subseteq E_i$ such that
\begin{enumerate}
	\item $\deg_X(v)=d(v)$ for every $v \in B_i$,
	\item $\deg_X(v)=2$ for every $v \in V_i \setminus B_i$,
	\item $\omega(X)=w$,
	\item $X$ has no cycles, unless $d(v)=1$ for all $v \in B_i$.
\end{enumerate}
Here we let $\deg_X(v)$ denote the number of edges in $X$ that are incident to $v$.
Define $\mathcal{T}_i[d,w,M]$ as the family of edge sets $X \in \mathcal{T}_i[d,w]$ such that $X \cup M$ is a single cycle.
The dynamic programming table entries $t_i[d,w,M]$ computed in~\cite[Section 4]{CyganKN2018} are defined as the parity of $|\mathcal{T}_i[d,w,M]|$. The algorithm from~\cite{CyganKN2018} shows how to compute $t_i$ whenever $i$ is a leaf, introduce vertex, introduce edge, or forget bag, based on the table $t_j$ where $j$ is a child of $i$ in $\mathbb{T}$ (if $i$ is not a leaf bag). It remains to show how to compute $t_j$ based on the tables $t_{j}$ and $t_{j'}$ if $x$ is a join bag with children $j$ and $j'$. To this end, we provide a formula for this that sets up the tensor that we need to study. 
We also use the shorthand notation $\overline{M}$ to denote $\mathcal{B}(V(M),\overline{a})$, if $M = \mathcal{B}(V(M),a)$ and the vertex set $V(M)$ denotes the endpoints of $M$.

\begin{Lem}\label{lem:join}
	If $i$ is a join bag with children $j$ and $j'$, then
	\[
		t_i[d,w,M] \equiv_2 \sum_{\substack{d_j+d_{j'}=d \\ w_j+w_{j'}=w \\ \overline{M_j} \cup \overline{M_{j'}} \cup M \text{ is a cycle}}}t_j[d_j,w_j,M_j]\cdot t_{j'}[d_{j'},w_{j'},M_{j'}].
	\]
\end{Lem}	
\begin{proof}
	It is easy to see that $\mathcal{T}_i[d,w,M]$ equals
	\[
		\bigcup_{\substack{d_j+d_{j'}=d \\ w_j+w_{j'}=w}} \bigg\{ Y \cup Z\ \bigg|\ Y \in \mathcal{T}_j[d_j,w_j], Z\in \mathcal{T}_{j'}[d_{j'},w_{j'}],  Y \cup Z \cup M \text{ is a cycle}\bigg\},
	\]
	and all terms in the union are disjoint since $E_{j} \cap E_{j'} = \emptyset$. Hence, we have that $t_i[d,w,M]$ equals
	\begin{align*}
		&\hphantom{=} \sum_{\substack{d_j+d_{j'}=d \\ w_j+w_{j'}=w \\ Y \in \mathcal{T}_j[d_j,w_j] \\ Z\in \mathcal{T}_{j'}[d_{j'},w_{j'}]}} [Y \cup Z \cup M \text{ is a cycle}],\\
		\intertext{which we can rewrite by applying Lemma~\ref{lem:fact} with $Y$ and $Z \cup M$ (with degree $2$ vertices contracted), since they are both perfect matchings of $K_{d^{-1}_j(1)}$, into}
		&= \sum_{\substack{d_j+d_{j'}=d \\ w_j+w_{j'}=w \\ Z\in \mathcal{T}_{j'}[d_{j'},w_{j'}]}} \sum_{a \in \{0,1\}^{|d_j^{-1}(1)|/2-1}} t_j[d_j,w_j,\mathcal{B}(d_j^{-1}(1),a)] \cdot [\mathcal{B}(d_j^{-1}(1),\overline{a}) \cup Z \cup M \text{ is a cycle}],\\
		\intertext{which we can rewrite by applying Lemma~\ref{lem:fact} with $Z$ and $\mathcal{B}(d_j^{-1}(1),a) \cup M$ (with degree $2$ vertices contracted), since they are both perfect matchings of $K_{d^{-1}_{j'}(1)}$, into}
		&= \sum_{\substack{d_j+d_{j'}=d \\ w_j+w_{j'}=w \\ Z\in \mathcal{T}_{j'}[d_{j'},w_{j'}]}} \sum_{\substack{a \in \{0,1\}^{|d_j^{-1}(1)|/2-1}\\b \in \{0,1\}^{|d_{j'}^{-1}(1)}|/2-1}} t_j[d_j,w_j,\mathcal{B}(d_j^{-1}(1),a)] \cdot t_{j'}[d_{j'},w_{j'},\mathcal{B}(d_{j'}^{-1}(1),b)]\\
		&\hspace{14em} \cdot [\mathcal{B}(d_j^{-1}(1),\overline{a}) \cup \mathcal{B}(d_{j'}^{-1}(1),\overline{b}) \cup M \text{ is a cycle}].
	\end{align*}
\end{proof}

\subsection{Kronecker scaling for the Matchings Connectivity Tensor}\label{subsec:mct}
Partition $U$ into $r = \lceil q/b \rceil $ blocks $U_1,\ldots,U_{r}$ of size at most $b$.
Let $(d_1,M_1),(d_2,M_2),(d_3,M_3)$ be a triple consisting of three $U$-fingerprints. The \emph{type} of this triple is defined as the triple $(X_1,X_2,X_3)$ where $X_i \subseteq M_i$ is the subset edges of $M_i$ that have both endpoints in distinct blocks for $i=1,2,3$ and $M^*$ is obtained from $M_1 \cup M_2 \cup M_3$ by contracting all vertices that are not an endpoint of an edge in $X_1 \cup X_2 \cup X_3$.
We let $T^q_b$ denote the set of all types. We first show this set is relatively small, using the following easy observation about the family of basis matchings:

\begin{Obs}\label{obs:cut}
	For any $i \in [t]$ and $X \subseteq [t]$ and perfect matching $M$ of $Z_X$, there are at most $2$ edges in $Z_X$ with one endpoint in $\{x_j: j\leq i\}$ and one endpoint in $\{x_j: j > i\}$.
\end{Obs}

\begin{Lem}
	For positive integers $b < q$ we have that
	\[
		|T^q_b| \leq (20b)^{12r}.
	\]
\end{Lem}
\begin{proof}
	For every $i \in \{1,2,3\}$ and $j \in \{1,\ldots,r\}$, let $x_{i,j}$ be the number of edges in $X_i$ with exactly one endpoint in $B_j$, and let $U_{i,j}$ be the set of vertices in $U_j$ incident to an edge of $X_i$.
	
	We have $|U_{i,j}| \leq x_{i,j}$, and by Observation~\ref{obs:cut}, $x_{i,j} \leq 4$.
	Moreover, if a vertex in $U_{i,j}$ is matched to a vertex in $U_{i,j'}$ in $X_i$, then as a direct consequence of the definition of the basis matchings from Definition~\ref{def:basis} there is at most one $j''$ such that $U_{i,j''}$ is nonempty.
	
	Hence, we can describe $X_i$ with $U_{i,1},\ldots,U_{i,r}$ and per vertex in $U_{i,1},\ldots,U_{i,r}$ there are at most $19$  possible\footnote{We arrived at this rough upper bound by counting five possible blocks of each $4$ vertices minus the vertex itself.} vertices to whom it could be adjacent in $X_i$. Hence, the number of possibilities for $X_i$ is at most
	\[
		\binom{b}{4}^r 20^{4r}	\leq (20b)^{4r}.
	\]
	Hence the number of options for $(X_1,X_2,X_3)$ is as claimed in the lemma statement.
\end{proof}

Since there are only few options for $(X_1,X_2,X_3)$ we can sum over all possibilities and deal with each one separately. Unfortunately this does not directly help to decompose $H_q$ into a Kronecker product since the (at most $12$) edges leaving a block still cause complications. We now argue this can be reduced to two edges leaving the block to a larger block (for blocks $U_1,\ldots,U_{r-1}$) and two edges leaving the block to a smaller block (for blocks $U_2,\ldots,U_{r}$) by decomposing $X_1 \cup X_2 \cup X_3$ into basis matchings.

Formally, suppose that $(d_1,M_1),(d_2,M_2),(d_3,M_3)$ are $U$-fingerprints such that $d_1(v)+d_2(v)+d_3(v)=2$, and suppose that the type of this triple is $\tau = (X_1,X_2,X_3)$, let $E(\tau)$ denote $X_1 \cup X_2 \cup X_3$, and let $V(\tau)$ denote all endpoints of $E(\tau)$.
Let $M^*$ be obtained by contracting all edges from $(M_1 \setminus X_1) \cup (M_2 \cup X_2) \cup (M_3 \cup X_3)$ that are not a self-loop. We have that $M_1 \cup M_2 \cup M_3$ is a single cycle if and only if $M^* \cup X$ is a single cycle. Hence, by Lemma~\ref{lem:fact}
\begin{equation}\label{eq:factx}
	[M_1 \cup M_2 \cup M_3 \text{ is a cycle}] \equiv_2 
	\sum_{\substack{ a \in \{0,1\}^{|V(\tau)|/2-1} \\ E(\tau) \cup \mathcal{B}(V(\tau),\overline{a}) \text{ is a HC }}} [M^* \cup \mathcal{B}(V(\tau),a) \text{ is a HC }].
\end{equation}
If we let $\mathcal{B}(\tau)$ denote 
\[
	\{ \mathcal{B}(V(\tau),a): E(\tau) \cup \mathcal{B}(V(\tau),\overline{a}) \text{ is a HC}\},
\]
then we can rewrite~\eqref{eq:factx} as
\begin{equation}\label{eq:factxb}
	[M_1 \cup M_2 \cup M_3 \text{ is a cycle}] \equiv_2 
	\sum_{A \in \mathcal{B}(\tau)} [M^* \cup A \text{ is a HC }].
\end{equation}

Moreover, since $A \in \mathcal{B}_\tau$, we have by Observation~\ref{obs:cut} for every $j=1,\ldots,r$ that at most two edges of $A$ leave $U_j$ to a block $U_{j'}$ with $j' < j$ and at most two edges of $A$ leave $U_j$ to a block $U_{j'}$ with $j' > j$. Call the first type of edges the \emph{left $A$-exits} of $U_j$ and the second type of edges the \emph{right $A$-exits} of $U_j$.

For $i\in\{1,2,3\}$, $j=1,\ldots,r$, $\tau \in T_b^q$, $A \in \mathcal{B}_\tau$ and $U$-fingerprints $(d_{1},M_{1}),(d_{2},M_{2}),(d_{3},M_{3})$ we define $U_j$-fingerprints
$\varphi(i,j,\tau,A,d_i,M_i):= (d^{\tau,A}_{i,j},M^{\tau,A}_{i,j})$ as follows:

$M^{\tau,A}_{1,j}$ is constructed from starting with all edges in $M_1$ contained in $U_j$ and by adding every edge in $A$ that is contained in $U_j$. Additionally:
\begin{itemize}
	\item If there is one edge in $A$ with endpoints in blocks $U_{j'}$ and $U_{j''}$ with $j'< j < j''$, there is at most one left exit of $U_j$ and at most one right exit of $U_j$. Add the endpoints of those exits that are in $U_j$ to $M^{\tau,A}_{1,j}$.

	\item If there are two left exits, add the endpoints of those exits that are in $U_j$ to $M^{\tau,A}_{1,j}$.
	
	\item If there are two right exits, add the endpoints of those exits that are in $U_j$ to $M^{\tau,A}_{1,j}$.	
\end{itemize} 

$d^{\tau,A}_{1,j}$ is defined as $d^{\tau,A}_{1,j}(V(M^{\tau,A}_{1,j}))=1$, and for all vertices in $U_j \setminus V(M^{\tau,A}_{1,j})$ we have $d^{\tau,A}_{1,j}(v)=d_1(v)$.

$M^{\tau,A}_{2,j}$ (respectively,$M^{\tau,A}_{3,j}$) are defined as $M_2$ (respectively, $M_3$) restricted to all edges contained in $U_j$ and $d^{\tau,A}_{2,j}$ (respectively, $d^{\tau,A}_{3,j}$) are defined as $d_2$ (respectively, $d_3$) restricted to $U_j$ with the exception that edges in $A$ do not contribute anymore.

Though somewhat tediously, we can make the following observation about this rerouting:

\begin{Obs}\label{obs:reroute}
	$M^* \cup A$ is a single cycle if and only if for every $j=1,\ldots,r$ we have that  $M^{\tau,A}_{1,j} \cup M^{\tau,A}_{2,j} \cup M^{\tau,A}_{3,j}$ is a single cycle.
\end{Obs} 

Hence, we can define
\[
\begin{aligned}
x^{\tau,A}_{(d_1,M_1),\ldots,(d_r,M_r)} &= \sum_{\substack{ (d,M) \\ \forall j \in \{1,\ldots,r\}: \varphi(1,j,\tau,A,d,M)=(d_j,M_j)}}x_{d,M},\\
y^{\tau,A}_{(d_1,M_1),\ldots,(d_r,M_r)} &= \sum_{\substack{ (d,M) \\ \forall j \in \{1,\ldots,r\}: \varphi(2,j,\tau,A,d,M)=(d_j,M_j)}}y_{d,M},\\
z^{\tau,A}_{(d_1,M_1),\ldots,(d_r,M_r)} &= \sum_{\substack{ (d,M) \\ \forall j \in \{1,\ldots,r\}: \varphi(3,j,\tau,A,d,M)=(d_j,M_j)}}z_{d,M}.
\end{aligned}
\]
and obtain the main result of this subsection:
\begin{Thm}[Kronecker scaling for the Matchings Connectivity Tensor]\label{thm:fac}
\label{thm:matchings-connectivity-tensor-kronecker-scaling}
	For all $q$ and $b$:
	\[
		H_q(x,y,z) = \sum_{\tau\in T_b^{q}} \sum_{\substack{A \in \mathcal{B}_\tau \\ E(\tau) \cup A \textnormal{ is a cycle}}} \left(\bigotimes_{j\in [r]} H_{|U_j|} \right)(x^{\tau,A},y^{\tau,A},z^{\tau,A}).
	\]
\end{Thm}
\begin{proof}
	We have by definition that
	\begin{align*}
		H_q(x,y,z) &\equiv_2 \sum_{\substack{(d_1,M_1),(d_2,M_2),(d_3,M_3) \\ \forall v \in [q]: d_1(v)+d_2(v)=d_3(v)} }x_{d_1,M_1}\cdot y_{d_2,M_2}\cdot z_{d_3,M_3} [M_1 \cup M_2 \cup M_3 \text{ is a cycle}],\\
				   &\equiv_2 \sum_{\substack{(d_1,M_1),(d_2,M_2),(d_3,M_3) \\ \forall v \in [q]: d_1(v)+d_2(v)=d_3(v)}}x_{d_1,M_1}\cdot y_{d_2,M_2}\cdot z_{d_3,M_3} \sum_{\substack{ A \in \mathcal{B}_\tau \\ E(\tau) \cup A \text{ is a HC }}} [M^* \cup A \text{ is a HC }]\\
				   &\equiv_2  \sum_{\tau \in T^q_b}\sum_{\substack{ A \in \mathcal{B}_\tau \\ E(\tau) \cup A \text{ is a HC }}}\sum_{\substack{(d_1,M_1),(d_2,M_2),(d_3,M_3) \\ \forall v \in [q]: d_1(v)+d_2(v)=d_3(v)}}x_{d_1,M_1}\cdot y_{d_2,M_2}\cdot z_{d_3,M_3}\cdot [M^* \cup A \text{ is a HC }]\\				   
   				   &\equiv_2  \sum_{\tau \in T^q_b}\sum_{\substack{ A \in \mathcal{B}_\tau \\ E(\tau) \cup A \text{ is a HC }}}\left(\bigotimes_{j\in [r]} H_{|U_j|} \right)(x^{\tau,A},y^{\tau,A},z^{\tau,A}),
	\end{align*}
	where the second congruence is by~\eqref{eq:factxb}, and $\tau$ denotes the type of $(d_1,M_1),(d_2,M_2),(d_3,M_3)$, the third equivalence is by reordering summations (and hence the third summation runs over all $(d_1,M_1),(d_2,M_2),(d_3,M_3)$ with type $\tau$), and the final equivalence is by Observation~\ref{obs:reroute}.
\end{proof}

\subsection{Rank bounds imply faster algorithms for Hamiltonicity parameterized by treewidth}
\label{subsec:fin}
We now prove Theorem~\ref{thm:hcmain}. 
By Lemma~\ref{lem:join} and the discussion in Subsection~\ref{subsec:revckn}, it suffices to compute
\[
t_i[d,w,M] \equiv_2 \sum_{\substack{d_j+d_{j'}=d \\ w_j+w_{j'}=w \\ \overline{M_j} \cup \overline{M_{j'}} \cup M \text{ is a cycle}}}t_j[d_j,w_j,M_j]\cdot t_{j'}[d_{j'},w_{j'},M_{j'}],
\] 
in $O^*\left((2+\sqrt{2}+o(1))^{(\sigma(H_\mathbb{N})+\varepsilon)\tw}\right)$ time, when we are given tables $t_j$ and $t_{j}$.
We use $H_d$ as a bilinear form via its partial derivatives. We iterate over all $w_j,w_{j'}$ such that $w_j+w_{j'}=w$ and let
\[
	y_{d,M} := t_j[d,w_j,\overline{M_j}], \quad z_{d,M} := t_{j'}[d_{j'},w_{j'},\overline{M_{j'}}].
\]
Consider the polynomial $H_{\tw}(x,y,z)$ in variables $x$.
By Theorem~\ref{thm:fac} we have that
\[
	H_\tw(x,y,z) = \sum_{\tau\in T_b^{\tw}} \sum_{\substack{A \in \mathcal{B}_\tau \\ E(\tau) \cup A \textnormal{ is a cycle}}} \left(\bigotimes_{j\in [r]} H_{|U_j|} \right)(x^{\tau,A},y^{\tau,A},z^{\tau,A}),
\]
and it is easy to see from the proof of Theorem~\ref{thm:fac} that $T^\tw_b,B_\tau,y^{\tau,A},z^{\tau,A}$ can all be constructed in $O^*\left((2+\sqrt{2})^{\tw}\right)$ time.

\begin{Lem}
	For every $\epsilon> 0$ we can, given $\tw$, $r$, $\tau$ and $A$, produce in $O((2+\sqrt{2})^{\sigma(H_{\mathbb{N}})+\varepsilon})$ time an arithmetic circuit $C$ of size $O((2+\sqrt{2})^{\sigma(H_{\mathbb{N}})+\varepsilon})$ that evaluates $\left(\bigotimes_{j\in [r]} H_{|U_j|} \right)(x^{\tau,A},y^{\tau,A},z^{\tau,A})$.
\end{Lem}
\begin{proof}
	Recall that $\sigma(H_{\mathbb{N}})=\inf\,\biggl\{\sigma>0:\Ra(H_q)\leq(2+\sqrt{2})^{q(\sigma+o(1))}\biggr\} $. Hence, for every $\epsilon$ there exists $b$ such that $\Ra(H_b) \leq (2+\sqrt{2})^{b(\sigma+\epsilon)}$.
	
	As mentioned below Definition~\ref{def:matchtensor}, $H_{q'}(x,y,z)$ is a sub-tensor of $H_{q}(x,y,z)$ whenever $q' < q$. Hence, $H_{|U_r|}$ is a sub-tensor $H_{b}$ and by adjusting 
	$x^{\tau,A},y^{\tau,A},z^{\tau,A}$ accordingly if needed, we can restrict attention to evaluating
	\[
		H_b^{\otimes r}(x^{\tau,A},y^{\tau,A},z^{\tau,A}).
	\]
	Applying Lemma~\ref{lem:yates-kron} with $T$ being the tensor $H_b$ and $r$ being the Kronecker power, we obtain the lemma statement.	
\end{proof}

Then, if we denote $\overline{d}$ for the vector $(2-d_1,\ldots,2-d_\tw)$, we have that
\[
	\frac{\partial H_{\tw}(x,y,z)}{\partial x_{\overline{d},M}} = \sum_{\substack{d_j+d_{j'}=d \\ \overline{M_j} \cup \overline{M_{j'}} \cup M \text{ is a cycle}}}t_j[d_j,w_j,M_j]\cdot t_{j'}[d_{j'},w_{j'},M_{j'}],
\]
and hence we can recover all values $t_i[d,w,M]$ efficiently once we computed all partial derivatives $\frac{\partial H_{\tw}(x,y,z)}{\partial x_{\overline{d},M}}$.
To do so, we use the following well-known lemma:
\begin{Lem}[Baur-Strassen~\cite{BaurS83}]\label{lem:bs}
	If a polynomial $P(x_1,\ldots,x_n) \in \mathbb{F}[x_1,\ldots,x_n]$ can be computed by an arithmetic circuit $C$ of size $s$, then there is another arithmetic circuit of size $O(s)$ that computes all partial derivatives $\frac{\partial P(x_1,\ldots,x_n)}{\partial x_1},\ldots,\frac{\partial P(x_1,\ldots,x_n)}{\partial x_n}$ simultaneously.
\end{Lem}
The circuit promised by Lemma~\ref{lem:bs} can also be constructed from $C$ in linear time, see~\cite{Morgenstern85}.
Thus, we can turn the arithmetic circuit $C$ constructed above into one that computes all partial derivates with linear overhead, as required.

\ifdefined\A
\else
\section*{Acknowledgments}
AB was supported by the VILLUM Foundation, Grant 54451. TK and JN were supported by the European Research Council (ERC) under the European Union's Horizon 2020 research and innovation programme grant agreement No 853234. TK was also supported by JSPS KAKENHI Grant Number JP20H05967.
\fi

%%%%%%%%%%%%%%%%%%%%%%%%%%%%%%%%%%%%%%%%%%%%%%%%%%%%%%%%%%%%%%%% References %%%

\bibliographystyle{abbrv}
\bibliography{paper}

%%%%%%%%%%%%%%%%%%%%%%%%%%%%%%%%%%%%%%%%%%%%%%%%%%%%%%%%%%%%% Document ends %%%

\end{document}